\documentclass[11pt,reqno]{amsart} 

\usepackage{times}

\usepackage{amssymb}
\usepackage{amsxtra}
\usepackage{amsmath}

\usepackage{verbatim}
\usepackage{epsfig}
\usepackage{setspace}

\newtheorem{theorem}{Theorem}[section]

\newtheorem{lemma}[theorem]{Lemma}
\newtheorem{proposition}[theorem]{Proposition}
\newtheorem{corollary}[theorem]{Corollary}
\newtheorem{example}{Example}[section]
\newtheorem{conjecture}{Conjecture}[section]

\def\k{\kappa}

\def\sg{\sigma}
\def\D{\Delta}
\def\G{\Gamma}

\def\C{{\mathbb C}}
\def\F{{\mathbb F}}

\def\a{{\mathbf a}}
\def\b{{\mathbf b}}
\def\c{{\mathbf c}}
\def\d{{\mathbf d}}
\def\g{{\mathbf g}}

\def\s{{\mathbf s}}

\def\w{{\mathbf w}}
\def\x{{\mathbf x}}
\def\y{{\mathbf y}}
\def\z{{\mathbf z}}
\def\0{{\mathbf 0}}
\def\1{{\mathbf 1}}
\def\O{{\mathbf O}}
\def\cC{{\mathcal C}}

\def\cG{{\mathcal G}}

\def\cM{{\mathcal M}}
\def\cP{{\mathcal P}}
\def\cQ{{\mathcal Q}}
\def\cS{{\mathcal S}}

\def\ob{\overline{\b}}
\def\tb{\widetilde{\b}}
\def\oF{{\overline{F}}}
\def\oJ{{\overline{J}}}
\def\oT{\overline{T}}

\def\he{{\hat{e}}}

\def\mfB{{\mathfrak B}}
\def\mfC{{\mathfrak C}}

\def\sfS{\mbox{\textsf{S}}}
\def\BW{\mbox{\textsf{BW}}}
\def\TW{\mbox{\textsf{TW}}}

\def\rank{\mbox{\textsf{rank}}}
\def\rref{\mbox{\textsf{rref}}}
\def\ess{\mbox{\textsf{ess}}}
\def\merge{\mbox{\textsf{merge}}}
\def\tr{{\text{trellis}}}
\def\path{{\text{path}}}

\def\bar{\overline}

\def\define{\stackrel{\mbox{\footnotesize def}}{=}}
\def\beq{\begin{equation}}
\def\eeq{\end{equation}}

\def\disj{\stackrel{\cdot}{\cup}}

\title{On Minimal Tree Realizations of Linear Codes}
\thanks{This work was supported by a Discovery Grant from the 
Natural Sciences and Engineering Research Council (NSERC), Canada.}
\author{Navin Kashyap} 

\thanks{The author is with the Department of Mathematics and Statistics,
Queen's University, Kingston, ON K7L 3N6, Canada.
Email: \texttt{nkashyap@mast.queensu.ca}}

\renewcommand{\markboth}[2]
{\renewcommand{\leftmark}{\scshape{#1}}\renewcommand{\rightmark}{\scshape{#2}}}

\begin{document}

\renewcommand{\thefootnote}{\arabic{footnote}}
\setcounter{footnote}{0}
\begin{abstract}
A tree decomposition of the coordinates of a code is a mapping 
from the coordinate set to the set of vertices of a tree. A tree 
decomposition can be extended to a tree realization, \emph{i.e.}, 
a cycle-free realization of the code on the underlying tree,
by specifying a state space at each edge of the tree, 
and a local constraint code at each vertex of the tree. 
The constraint complexity of a tree realization is 
the maximum dimension of any of its local constraint codes.
A measure of the complexity of maximum-likelihood decoding for a
code is its treewidth, which is the least constraint complexity 
of any of its tree realizations.

It is known that among all tree realizations of a code 
that extends a given tree decomposition, there exists a unique 
minimal realization that minimizes the state space dimension at 
each vertex of the underlying tree. In this paper, we give two 
new constructions of these minimal realizations. As a by-product
of the first construction, a generalization of the state-merging 
procedure for trellis realizations, we obtain the fact that the 
minimal tree realization also minimizes the local constraint code 
dimension at each vertex of the underlying tree. The second construction 
relies on certain code decomposition techniques that we develop.
We further observe that the treewidth of a code is related to a measure 
of graph complexity, also called treewidth. We exploit this connection
to resolve a conjecture of Forney's regarding the gap between the
minimum trellis constraint complexity and the treewidth of a code. 
We present a family of codes for which this gap can be arbitrarily large.
\end{abstract}

\date{\today}
\maketitle


\section{Introduction\label{intro}}

Graphical models of codes and the decoding algorithms associated
with them are now a major focus area of research in coding theory.
Turbo codes, low-density parity-check (LDPC) codes, and expander codes
are all examples of codes defined, in one way or another, on underlying 
graphs. A unified treatment of graphical models and the associated
decoding algorithms began with the work of Wiberg, Loeliger and Koetter 
\cite{wiberg},\cite{WLK95}, and has since been abstracted and refined 
under the framework of the generalized distributive law \cite{AM00},
factor graphs \cite{KFL01}, and normal realizations 
\cite{For01},\cite{For03}. The particular case of graphical models in which 
the underlying graphs are cycle-free has a long and rich history of its own, 
starting with the study of trellis representations of codes; 
see \emph{e.g.}, \cite{vardy} and the references therein. 

Briefly, a graphical model consists of a graph, an assignment of 
symbol variables to the vertices of the graph, an assignment of 
state variables to the edges of the graph, and a specification 
of local constraint codes at each vertex of the graph.
The full behavior of the model is the set of all configurations
of symbol and state variables that satisfy all the local constraints.
Such a model is called a realization of a code $\cC$ if the
restriction of the full behavior to the set of symbol variables 
is precisely $\cC$. The realization is said to be cycle-free if
the underlying graph in the model has no cycles. A trellis 
representation of a code can be viewed as a cycle-free realization 
in which the underlying graph is a simple path.

A linear code $\cC$ has a realization on a graph $\cG$ 
that is not connected if and only if $\cC$ can be expressed as 
the direct sum of the codes that are individually realized on the connected 
components of $\cG$ \cite{For01}. Thus, there is no loss of generality 
in just focusing, as we do, on the case of realizations on connected graphs.
In this paper, we will be concerned with tree realizations ---
cycle-free realizations in which the underlying cycle-free graph 
is connected, \emph{i.e.}, is a tree. 

It is by now well known that the sum-product algorithm on any 
tree realization provides an exact implementation of 
maximum-likelihood (ML) decoding 
\cite{AM00},\cite{For01},\cite{KFL01},\cite{wiberg}. 
A good initial estimate of the computational complexity of such
an implementation is given by the constraint complexity of the 
realization, which is the maximum dimension of any of the local 
constraint codes in the realization. Now, distinct tree realizations
of the same code have, in general, distinct constraint complexities.
The treewidth of a code is defined to be the least constraint 
complexity of any of its tree realizations. Thus, treewidth 
may be taken to be a measure of the ML decoding complexity of a code.

Since trellis realizations are instances of tree realizations,
the treewidth of a code can be no larger than the minimum 
constraint complexity\footnote{In the context of trellis realizations, 
constraint complexity is usually referred to as ``branch complexity''
or ``edge complexity''. We make it a point to avoid this usage, so 
as not to cause confusion when we define the ``branchwidth'' of a 
code later in our paper.} of any of its trellis realizations. 
In the abstract of his paper \cite{For03}, Forney claimed that 
``the constraint complexity of a general cycle-free graph 
realization can be [strictly] less than that of any
conventional trellis realization, but not by very much.'' While 
he substantiated the first part of his claim by means of an example,
he left the ``not by very much'' part as a conjecture 
\cite[p.\ 1606, Conjecture~2]{For03}. But he also admitted that none
of the arguments he gave in support of his conjecture 
``is very persuasive,'' and that it is equally plausible 
that \cite[Conjecture~3]{For03} there exists no upper bound on the 
gap between the treewidth of a code and the minimum constraint complexity
of any of its trellis realizations.

One of the main contributions of this paper is an example that
affirms the validity of Forney's Conjecture~3.
We present, in Section~\ref{complexity_section}, a family of codes 
for which the difference between the minimum trellis constraint 
complexity and the treewidth grows logarithmically with codelength.
We conjecture that this is in fact the maximal rate of growth of this
difference. Our construction of this example is based upon results 
from the graph theory and matroid theory literatures that connect the
notions of treewidth and trellis complexity of a code to certain
complexity measures defined for graphs.

This paper makes two other contributions, both relating to minimal tree
realizations. A mapping of the set of coordinates of a code $\cC$
to the vertices of a tree is called a tree decomposition. A tree
decomposition may be viewed as an assignment of symbol variables to the
vertices of the tree. It is known that given a code $\cC$, 
among all tree realizations of $\cC$ that extend a given 
tree decomposition, there is one that minimizes the state space 
dimension at each vertex of the underlying tree \cite{For01}. 
This minimal tree realization, an explicit construction of 
which was also given in \cite{For01}, is unique up to isomorphism. 

We give two new constructions of minimal tree realizations. 
The first construction involves a generalization of the idea 
of state merging that can be used to construct minimal trellis 
realizations \cite[Section~4]{vardy}. We show that any tree 
realization of a code can be converted to a minimal realization
by a sequence of state merging transformations. The state space 
and constraint code dimensions do not increase at any step of 
this process. From this, we obtain the fact that a minimal 
realization also minimizes the constraint code dimension at 
each vertex of the underlying tree.

Our second construction of minimal tree realizations uses extensions
of the code decomposition techniques that were presented in 
\cite{kashyap}. The main advantage of this construction is 
its recursive nature, which makes it suitable for mechanical
implementation. Also, it is relatively straightforward to 
estimate the computational complexity of this construction. 
We show that the complexity is polynomial in the length and 
dimension of the code, as well as in the size of the underlying tree, 
but is exponential in the state-complexity of the minimal realization, 
which is the maximum dimension of any state space in the realization.

The paper is organized as follows. In Section~\ref{background_section}, 
we provide the necessary background on tree realizations of linear codes.
The construction of minimal realizations by means of state merging 
is presented in Section~\ref{state_merging_section}. Code decomposition
techniques are developed in Section~\ref{decomp_section}, and used in 
Section~\ref{new_construct_section} to derive a recursive construction
of minimal tree realizations. Proofs of some of the results from 
Sections~\ref{background_section}--\ref{new_construct_section} are 
deferred to appendices to preserve the flow of the exposition.
Treewidth and related complexity measures are defined in 
Section~\ref{complexity_section}, which also establishes 
connections between these code complexity measures and 
certain complexity measures defined for graphs. These connections
are used to derive the example of a code family for which 
the gap between minimum trellis constraint complexity and treewidth is
arbitrarily large. We also touch upon the subject of codes of
bounded complexity, observing that many hard coding-theoretic problems 
become polynomial-time solvable when restricted to code families
whose treewidth is bounded. Section~\ref{conclusion} contains a few
concluding remarks.

\section{Background on Tree Realizations\label{background_section}}

Our treatment of the topic of tree realizations in this section
is based on the exposition of Forney \cite{For01},\cite{For03};
see also \cite{halford}.

We start by establishing some basic notation. We take $\F$ to be 
an arbitrary finite field. Given a finite index set $I$, we 
have the vector space $\F^I = \{\x = (x_i \in \F,\ i \in I)\}$.
For $\x \in \F^I$ and $J \subseteq I$, the notation ${\x|}_J$ will
denote the \emph{projection} $(x_i,\ i \in J)$. Also, for $J \subseteq I$, 
we will find it convenient to reserve the use of $\bar{J}$ 
to denote the set difference $I - J = \{i \in I:\ i \notin J\}$.

\subsection{Codes\label{code_section}}
A \emph{linear code} over $\F$, defined on the index set $I$, is a 
subspace $\cC \subseteq \F^I$. We will only consider linear codes
in this paper, so the terms ``code'' and ``linear code'' will be used
interchangeably. The dimension, over $\F$, of $\cC$ will be denoted 
by $\dim(\cC)$. An $[n,k]$ code is a code of length $n$ and dimension $k$. 
If, additionally, the code has minimum distance $d$,
then the code is an $[n,k,d]$ code. The dual code of $\cC$ is denoted by
$\cC^\perp$, and is defined on the same index set as $\cC$.

Let $J$ be a subset of the index set $I$. The \emph{projection} of $\cC$ 
onto $J$ is the code $\cC|_J = \{{\c|}_J:\ \c \in \cC\}$, which is a 
subspace of $\F^J$. We will use $\cC_J$ to denote the \emph{cross-section}
of $\cC$ consisting of all projections $\c|_J$ of codewords $\c \in \cC$ 
that satisfy ${\c|}_{\bar{J}} = \0$. To be precise, 
$\cC_J = \{\c|_J:\ \c \in \cC, {\c|}_{\bar{J}} = \0\}$. Note that
$\cC_J \subseteq \cC|_J$. Also, since $\cC_J$ is isomorphic to
the kernel of the projection map $\pi: \cC \rightarrow {\cC|}_\oJ$
defined by $\pi(\c) = {\c|}_\oJ$, we have that 
$\dim(\cC_J) = \dim(\cC) - \dim({\cC|}_\oJ)$.
Furthermore, projections and cross-sections
are dual notions, in the sense that $({\cC|}_J)^\perp = (\cC^\perp)_J$,
and similarly, $(\cC_J)^\perp = {(\cC^\perp)|}_J$. 

If $\cC_1$ and $\cC_2$ are codes over $\F$ defined on mutually disjoint 
index sets $I_1$ and $I_2$, respectively, then their \emph{direct sum} is the
code $\cC = \cC_1 \oplus \cC_2$ defined on the index set $I_1 \cup I_2$,
such that $\cC_{I_1} = {\cC|}_{I_1} = \cC_1$ and 
$\cC_{I_2} = {\cC|}_{I_2} = \cC_2$. This definition naturally extends
to multiple codes (or subspaces) $\cC_\alpha$, where $\alpha$ is a
code identifier that takes values in some set $A$. Again, it must be assumed
that the codes $\cC_\alpha$ are defined on mutually disjoint index sets
$I_\alpha,\ \alpha \in A$. The direct sum in this situation
is denoted by $\bigoplus_{\alpha \in A} \cC_\alpha$. 

\subsection{Trees\label{tree_section}}

A tree is a connected graph without cycles. Given a tree $T$, 
we will denote its vertex and edge sets by $V(T)$ and $E(T)$, 
respectively, or simply by $V$ and $E$ if there is no ambiguity. 
Vertices of degree one are called \emph{leaves}, 
and all other vertices are called \emph{internal nodes}.
Given a $v \in V$, the set of edges incident with $v$ will be denoted
by $E(v)$.

Removal of an arbitrary edge $e$ from $T$ produces a disconnected
graph $T-e$, which is the disjoint union of two subtrees,
which we will denote by $T_e$ and $\bar{T}_e$, of $T$. Note that
$V(T_e)$ and $V(\bar{T}_e)$ form a partition of $V(T)$.

\subsection{Tree Realizations\label{realize_section}}

Let $\cC$ be a code over $\F$, defined on the index set $I$.
To each $i \in I$, we associate a \emph{symbol variable} $X_i$, 
which is allowed to take values in $\F$. 

\begin{figure}
\epsfig{file=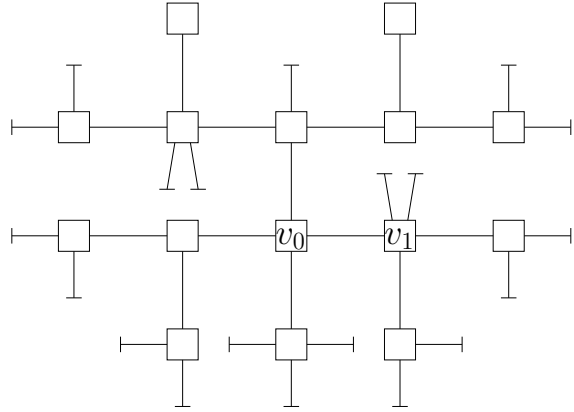, width=7.5cm}
\caption{A depiction of a tree decomposition $(T,\omega)$. Vertices
of the tree $T$ are represented by squares. Edges are incident 
with two vertices, while half-edges are incident with only one vertex. 
The vertex $v_0$ has no half-edges incident with it, indicating that 
$\omega^{-1}(v_0) = \emptyset$, while the vertex $v_1$ has two 
half-edges incident with it, which means that $|\omega^{-1}(v_1)| = 2$.}
\label{tree_decomp}
\end{figure}

A \emph{tree decomposition} of $I$ is a pair $(T,\omega)$, where 
$T$ is a tree (\emph{i.e.}, a connected, cycle-free graph) and 
$\omega: I \rightarrow V$ is a mapping from $I$ to the vertex 
set of $T$. Pictorially, a tree decomposition $(T,\omega)$ is depicted as a 
tree with an additional feature: at each vertex $v$ such that 
$\omega^{-1}(v)$ is non-empty, we attach special ``half-edges'', one for 
each index in $\omega^{-1}(v)$; see Figure~\ref{tree_decomp}.

At this point, we introduce some notation that we will consistently
use in the rest of the paper. Given a tree decomposition $(T,\omega)$
of an index set $I$, and an edge $e \in E$, we define
$J(e) = \omega^{-1}(V(T_e))$ and $\bar{J}(e) = \omega^{-1}(V(\bar{T}_e))$.
Thus, $J(e)$ and $\bar{J}(e)$ are the subsets of $I$
that get mapped by $\omega$ to vertices in $T_e$ and $\bar{T}_e$, 
respectively. Clearly, $J(e)$ and $\bar{J}(e)$ form a partition of $I$.

Recall that $E(v)$, $v \in V$, denotes the set of edges incident with $v$
in $T$. Consider a tuple of the form 
$(T,\omega,(\cS_e,\ e \in E), (C_v,\ v \in V))$, where
\begin{itemize}
\item $(T,\omega)$ is a tree decomposition of $I$;
\item for each $e \in E$, $\cS_e$ is a vector space over $\F$ called
a \emph{state space};
\item for each $v \in V$, $C_v$ is a subspace of 
$\F^{\omega^{-1}(v)} \oplus\, \left(\bigoplus_{e \in E(v)} \cS_e\right)$,
called a \emph{local constraint code}, or simply, a \emph{local constraint}.
\end{itemize}
Such a tuple will be called a \emph{tree model}. The elements of 
any state space $\cS_e$ are called \emph{states}. The index sets of the
state spaces $\cS_e$, $e \in E$, are taken to be mutually disjoint,
and are also taken to be disjoint from the index set $I$ corresponding to
the symbol variables. Finally, to each $e \in E$, we associate a 
\emph{state variable} $S_e$ that takes values in the corresponding 
state space $\cS_e$.

A \emph{global configuration} of a tree model as above is
an assignment of values to each of the symbol and state variables.
In other words, it is a vector of the form 
$((x_i \in \F,\ i \in I), (\s_e \in \cS_e,\ e \in E))$.
A global configuration is said to be \emph{valid} if it satisfies all
the local constraints. Thus, 
$((x_i \in \F,\ i \in I), (\s_e \in \cS_e,\ e \in E))$ is a valid
global configuration if for each $v \in V$, $((x_i,\ i \in \omega^{-1}(v)),
(\s_e,\ e \in E(v))) \in C_v$. The set of all valid global configurations
of a tree model is called the \emph{full behavior} of the model.

Note that the full behavior is a subspace 
$\mfB \subseteq \F^I \oplus\, \left(\bigoplus_{e \in E} \cS_e\right)$.
As usual, ${\mfB|}_I$ denotes the projection of $\mfB$ onto the
index set $I$. If ${\mfB|}_I = \cC$, then the model 
$(T,\omega,(\cS_e,\ e \in E), (C_v,\ v \in V))$ is called
a \emph{(linear) tree realization} of $\cC$. A tree realization 
$(T,\omega,(\cS_e,\ e \in E), (C_v,\ v \in V))$ of $\cC$
is said to \emph{extend} (or be an extension of) the tree decomposition 
$(T,\omega)$ of the index set of $\cC$. Any tree decomposition of the index
set of a code can always be extended to a tree realization of 
the code, as explained in the following example.

\begin{example} 
Let $\cC$ be a code defined on index set $I$, and let $(T,\omega)$ be
a tree decomposition of $I$. Pick an arbitrary $v \in V$, 
and define $C_v = \cC$. Now, consider the set, $E(v)$, of edges 
incident with $v$. Removal of any $e \in E(v)$ produces the
two subtrees $T_e$ and $\bar{T}_e$. We specify $T_e$ to be
the subtree that does \emph{not} contain the vertex $v$, and
as usual, $J(e) = \omega^{-1}(V(T_e))$. For each
$e \in E(v)$, the state space $\cS_e$ is taken to be a copy of $\F^{J(e)}$. 
The remaining state spaces and local constraints are chosen so that, for each
$e \in E(v)$, the symbol variables indexed by $J(e)$ simply get 
relayed (unchanged) to the state variable $S_e$; see 
Figure~\ref{trivial_ext}. It should be clear that the resulting 
tree model is a tree realization of the code $\cC$.
This will be called a \emph{trivial extension} of $(T,\omega)$. 
We will present constructions of non-trivial extensions of 
tree decompositions a little later.
\begin{figure}
\epsfig{file=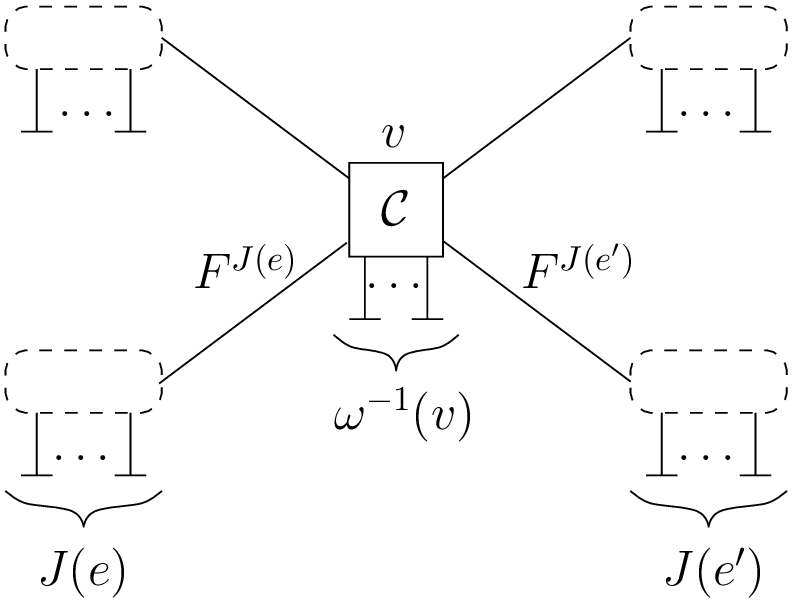, width=6cm}
\caption{A trivial extension of a tree decomposition 
$(T,\omega)$ of the index set of a code $\cC$. At the vertex $v$,
we have $C_v = \cC$. The state variables at the edges $e \in E(v)$ 
are copies of the symbol variables indexed by $J(e)$. Dashed ovals 
represent subtrees.}
\label{trivial_ext}
\end{figure}
\label{trivial_example}
\end{example}

\begin{figure}[h]
\epsfig{file=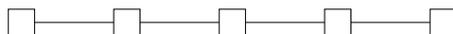, width=6cm}
\caption{A simple path on five vertices.}
\label{path_fig}
\end{figure}
\begin{example}
A \emph{simple path} is a tree with exactly two leaves (the end-points
of the path), in which all internal nodes have degree two; see 
Figure~\ref{path_fig}. 
Let $\cC$ be a code defined on index set $I$, and let $(T,\omega)$ be
a tree decomposition of $I$, in which $T$ is a simple path, and $\omega$
is a surjective map $\omega:I \rightarrow V(T)$. Any tree realization of 
$\cC$ that extends $(T,\omega)$ is called a \emph{trellis realization} of 
$\cC$. When $\omega$ is a bijection, then any trellis realization extending
$(T,\omega)$ is called a \emph{conventional} trellis realization. When
$\omega$ is not a bijection (but still a surjection), a trellis
realization that extends $(T,\omega)$ is called a \emph{sectionalized}
trellis realization. In trellis terminology, the local constraint codes
in a trellis realization are called \emph{branch spaces}. The theory
of trellis realizations is well established; we refer the reader to
\cite{vardy} for an excellent survey of this theory.
\label{trellis_example}
\end{example}

Let $\mfB$ be the full behavior of a tree model 
$(T,\omega,(\cS_e,\ e \in E), (C_v,\ v \in V))$.
We will find it useful to define certain projections of $\mfB$, 
other than $\mfB|_J$ for $J \subseteq I$.
Let $\b = ((x_i,\, i \in I),\ (\s_e,\, e \in E))$
be a global configuration in $\mfB$. At any given $v \in V$,
the \emph{local configuration} of $\b$ at $v$ is defined as
$$
{\b|}_v = ((x_i,\, i \in \omega^{-1}(v)),\ (\s_e,\, e \in E(v))).
$$
The set of all local configurations of $\mfB$ at $v$ is then defined as
${\mfB|}_v = \{{\b|}_v:\, \b \in \mfB\}$. By definition, 
${\mfB|}_v \subseteq C_v$. Similarly, for $F \subseteq E$, and $\b$
as above, we define the projections ${\b |}_F =(\s_e,\, e \in F)$ and 
${\mfB |}_F =\{{\b|}_F:\, \b \in \mfB\}$. Clearly, ${\mfB|}_F$ is a subspace
of $\bigoplus_{e \in F}\cS_e$. If $F$ consists of a single edge $e$, 
then we simply denote the corresponding projections by ${\b |}_e$
and ${\mfB |}_e$. The following elementary property of the 
projections ${\b|}_e$ will be useful later; a proof for it 
is given in Appendix~\ref{sec2_lemmas_app}.

\begin{lemma}
Let $\mfB$ be the full behavior of some tree realization of a code $\cC$,
defined on the index set $I$, that extends the tree decomposition 
$(T,\omega)$. Suppose that $\b \in \mfB$ and $e \in E$ are such that 
${\b|}_e = \0$. Then, ${\b|}_I \in \cC_{J(e)} \oplus \cC_{\bar{J}(e)}$.
\label{b|e_lemma}
\end{lemma}

A tree model (or realization) $(T,\omega,\,(\cS_e,\, e \in E),\, 
(C_v,\, v \in V))$, with full behavior $\mfB$, is said to be 
\emph{essential} if ${\mfB|}_e = \cS_e$ for all $e \in E$.
This definition actually implies something more. 

\begin{lemma}
If the tree model $(T,\omega,\,(\cS_e,\, e \in E),\, (C_v,\, v \in V))$, 
with full behavior $\mfB$, is essential, 
then ${\mfB|}_v = C_v\,$ for all $v \in V$. 
\label{ess_lemma}
\end{lemma}

A proof of the lemma can be found in Appendix~\ref{sec2_lemmas_app}. \\[-4pt]

An arbitrary tree model can always be ``essentialized''. To see this, let 
$\G = (T,\omega,\,(\cS_e,\, e \in E),\, (C_v,\, v \in V))$ be a 
tree model with full behavior $\mfB$. Recall that ${\mfB|}_e$ is
a subspace of $\cS_e$, and ${\mfB|}_v$ is a subspace of $C_v$.
Define the \emph{essentialization} of $\G$ to be the tree model
$\ess(\G) = (T,\omega,\,({\mfB|}_e,\, e \in E),\, ({\mfB|}_v,\, v \in V))$. 
It is readily verified that $\ess(\G)$ has the same full behavior as $\G$. 
\\[-6pt]


\subsection{Minimal Tree Realizations\label{minimal_section}}

Given a code $\cC$ and a tree decomposition $(T,\omega)$ of its index set 
$I$, there exists an essential tree realization, 
$(T,\omega,(\cS_e^*,\ e \in E), (C_v^*,\ v \in V))$, of $\cC$ 
with the following property \cite{For01},\cite{For03}: 
\begin{quote}
if $(T,\omega,(\cS_e,\ e \in E), (C_v,\ v \in V))$ is
a tree realization of $\cC$ that extends $(T,\omega)$, then 
for all $e \in E$, $\dim(\cS_e^*) \leq \dim(\cS_e)$.
\end{quote}
This \emph{minimal} tree realization, which we henceforth denote by 
$\cM(\cC; T,\omega)$, is unique up to isomorphism. More precisely, if 
$(T,\omega,(\cS_e^{**},\ e \in E), (C_v^{**},\ v \in V))$ is also 
a tree realization of $\cC$ with the above property
(except that $\cS_e^*$ is replaced by $\cS_e^{**}$), then 
$\cS_e^* \cong \cS_e^{**}$ for each $e \in E$, and $C_v^* \cong C_v^{**}$
for each $v \in V$. We will not distinguish between isomorphic tree
realizations.

We outline a construction, due to Forney \cite{For03}, of $\cM(\cC; T,\omega)$.
For any edge $e \in E$, the sets $J(e)$ and $\bar{J}(e)$ form a partition
of the index set $I$. Set 
\beq
\cS_e^* = \cC / (\cC_{J(e)} \oplus \cC_{\bar{J}(e)}),
\label{Se*_def}
\eeq
and let 
\beq
\s_e^*:\ \cC \rightarrow \cC / (\cC_{J(e)} \oplus \cC_{\bar{J}(e)})
\label{se*_def}
\eeq
be the canonical projection map. In other words, for $\c \in \cC$,
$\s_e^*(\c)$ is the coset $\c + (\cC_{J(e)} \oplus \cC_{\bar{J}(e)})$.

Now, let $\mfB$ be the vector space consisting of all global configurations
$(\c,\s^*(\c))$ corresponding to codewords $\c \in \cC$, where 
$\s^*(\c) = (\s_e^*(\c),\ e \in E)$. It is worth noting that 
${\mfB|}_I = \cC$, and furthermore, $\mfB \cong \cC$, since $\c = \0$ 
implies that $\s^*(\c) = \0$. 

We can now define for each $v \in V$, the local constraint
\beq
C_v^* = {\mfB |}_v 
=  \left\{\left({\c|}_{\omega^{-1}(v)}\, ,\ (\s_e^*(\c),\, e\in E(v))\right) \ : \ \c \in \cC\right\}.
\label{Cv*_def}
\eeq
The minimal realization $\cM(\cC; T,\omega)$ is the tuple
$(T,\omega,(\cS_e^*,\ e \in E), (C_v^*,\ v \in V))$. It may be verified
that $\mfB$ is the full behavior of $\cM(\cC;T,\omega)$, so that
$\cM(\cC;T,\omega)$ is indeed an essential tree realization of $\cC$.

From the definition of $\cS_e^*$ in (\ref{Se*_def}), it is clear that 
for each $e \in E$,
\beq
\dim(\cS_e^*) = \dim(\cC) - \dim(\cC_{J(e)}) - \dim(\cC_{\bar{J}(e)}).
\label{dimSe*}
\eeq
It is useful to point out that $\dim(\cS_e^*)$ may also be expressed
as 
\beq
\dim(\cS_e^*) = \dim({\cC|}_{J(e)}) + \dim({\cC|}_{\bar{J}(e)}) - \dim(\cC),
\label{dimSe*_alt}
\eeq
a consequence of the fact that for any $J \subseteq I$, 
$\dim(\cC_J) = \dim(\cC) - \dim({\cC|}_\oJ)$. 
Thus, by the uniqueness of minimal tree realizations,
if $\G^{**} = (T,\omega,(\cS_e^{**},\ e \in E), (C_v^{**},\ v \in V))$ is 
a tree realization of $\cC$ with the property that for all $e \in E$,
$\dim(\cS_e^{**})$ equals one of the expressions in 
(\ref{dimSe*}) or (\ref{dimSe*_alt}), then $\G^{**}$ is in fact 
$\cM(\cC;T,\omega)$.

Forney \cite{For03} also derived an expression for the dimension 
of the local constraints $C_v^*$. Consider any $v \in V$.
For each $e \in E(v)$, we specify $T_e$ to be the component of $T-e$
that does \emph{not} contain $v$. As usual, $J(e) = \omega^{-1}(V(T_e))$. 
Then \cite[Theorem 1]{For03},
\beq
\dim(C_v^*) = \dim(\cC) - \sum_{e \in E(v)} \dim(\cC_{J(e)}).
\label{dimCv*}
\eeq
Forney gave the following bound for $\dim(C_v^*)$ \cite[Theorem~5]{For03}: 
for any $e \in E(v)$, 
$\dim(\cS_e^*) \leq \dim(C_v^*) \leq n(C_v^*) - \dim(\cS_e^*)$, where
$n(C_v^*)$ denotes the length of the code $C_v^*$. The upper bound can be
improved slightly.

\begin{lemma}
In the minimal tree realization $\cM(\cC;T,\omega)$, we have,
for $v \in V$ and $e \in E(v)$, 
$$
\dim(\cS_e^*) \leq \dim(C_v^*) 
\leq \dim({\cC|}_{\omega^{-1}(v)}) + 
\sum_{e' \in E(v)- \{e\}} \dim(\cS_{e'}^*).
$$
\label{dimCv*_bnd}
\end{lemma}
\begin{proof}
The upper bound may be proved as follows. Since 
$\dim(\cC_J) = \dim(\cC) - \dim({\cC|}_\oJ)$ for any $J \subseteq I$,
we may write (\ref{dimCv*}) as 
$$
\dim(C_v^*) = \sum_{e \in E(v)} \dim({\cC|}_{\oJ(e)}) - (|E(v)-1) \dim(\cC).
$$
Now, let $e \in E(v)$ be fixed. We have 
$$
\dim(C_v^*) = \dim({\cC|}_{\oJ(e)}) + 
\sum_{e' \in E(v) - \{e\}} \left(\dim({\cC|}_{\oJ(e')}) - \dim(\cC)\right).
$$
However, as can be seen from Figure~\ref{vpersp_fig}, $\oJ(e)$ is the
disjoint union of $\omega^{-1}(v)$ and the sets $J(e')$, 
$e' \in E(v) - \{e\}$.
\begin{figure}
\epsfig{file=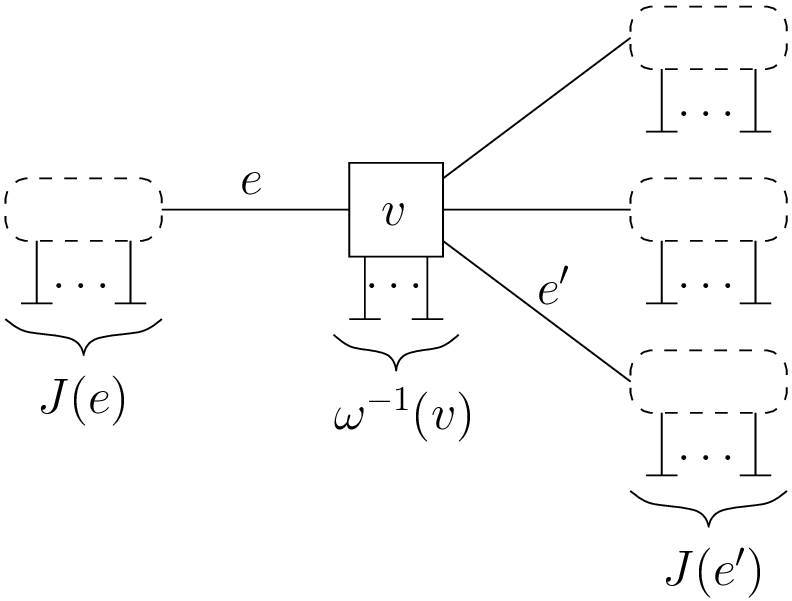, width=6cm}
\caption{Figure depicting the perspective about a vertex $v$ in
a tree decomposition. Dashed ovals represent subtrees.}
\label{vpersp_fig}
\end{figure}
Therefore,
$$
\dim({\cC|}_{\oJ(e)}) \leq \dim({\cC|}_{\omega^{-1}(v)}) + 
\sum_{e' \in E(v) - \{e\}} \dim({\cC|}_{J(e')}),
$$
and hence,
$$
\dim(C_v^*) \leq \dim({\cC|}_{\omega^{-1}(v)}) + 
\sum_{e' \in E(v) - \{e\}} 
\left( \dim({\cC|}_{J(e')}) + \dim({\cC|}_{\oJ(e')}) - \dim(\cC)\right).
$$ 
The lemma now follows from (\ref{dimSe*_alt}).
\end{proof}

As already mentioned, among all tree realizations of $\cC$
extending $(T,\omega)$, the minimal realization $\cM(\cC;T,\omega)$ 
minimizes state space dimension at each edge of the tree $T$. It is 
natural to ask whether $\cM(\cC;T,\omega)$ also minimizes local
constraint code dimension at each vertex of $T$. We will show in 
the next section that $\cM(\cC;T,\omega)$ does in fact have 
the following property:
\begin{quote}
if $(T,\omega,(\cS_e,\ e \in E), (C_v,\ v \in V))$ is
a tree realization of $\cC$ that extends $(T,\omega)$, then 
for all $v \in V$, $\dim(C_v^*) \leq \dim(C_v)$.
\end{quote}
We will deduce this fact from an alternative construction of
$\cM(\cC;T,\omega)$ that we present next.

\section{A Construction of $\cM(\cC;T,\omega)$ via 
State Merging\label{state_merging_section}}

The construction we describe in this section takes an arbitrary tree 
realization $\G$ that extends the tree decomposition $(T,\omega)$ --- 
for example, the trivial extension given in Example~\ref{trivial_example} --- 
and via a sequence of transformations, converts $\G$ to $\cM(\cC;T,\omega)$.
These transformations constitute a natural generalization of 
the state-merging process in the context of minimal trellis realizations;
see, for example, \cite[Section~4]{vardy}. It would be useful to keep 
this special case in mind while going through the details of the
description that follows.

Let $\G = (T,\omega,(\cS_e,\ e \in E), (C_v,\ v \in V))$ be an 
essential\footnote{This restriction can be dropped by considering
$\ess(\G)$ instead; see Theorem~\ref{merge_theorem}.} tree realization 
of a code $\cC$ with index set $I$, and let 
$\mfB$ be the full behavior of $\G$. As $\G$ is essential,
we have that ${\mfB|}_e = \cS_e$ for all $e \in E$ (by definition), and 
${\mfB|}_v = C_v$ for all $v \in V$ (by Lemma~\ref{ess_lemma}).

Pick an arbitrary edge $\he \in E$, and for ease of notation, set
$J = J(\he)$ and $\bar{J} = \bar{J}(\he)$. Let $W$ be the subspace
of $\cS_{\he}$ defined by
$$
W = \{\s \in \cS_{\he}:\ \exists\, \b \in \mfB \text{ such that }
{\b|}_I \in \cC_J \oplus \cC_{\bar{J}}\, \text{ and }\, {\b|}_{\he} = \s\}.
$$
We will define a new tree model $\bar{\G} = 
(T,\omega,(\bar{\cS}_e,\, e \in E), (\bar{C}_v,\, v \in V))$,
such that states in the same coset of $W$ in $\cS_{\he}$ are represented 
by a single ``merged'' state in $\bar{\cS}_{\he}$.


Let 
$$
\Phi: \F^I \oplus \left(\bigoplus_{e \neq \he} \cS_e\right) \oplus \, \cS_{\he}
\ \longrightarrow \ 
\F^I \oplus \left(\bigoplus_{e \neq \he} \cS_e\right) \oplus \, \cS_{\he}/W
$$
be the mapping defined by 
$$
\Phi((x_i,\, i \in I),\, (\s_e,\, e \neq \he),\, \s)
\ \ = \ \ ((x_i,\, i \in I),\, (\s_e,\, e \neq \he),\, \s+W).
$$
Define $\bar{\mfB} = \Phi(\mfB)$. It is clear from the definitions
that ${\bar{\mfB}|}_I = {\mfB|}_I = \cC$, and that 
$\dim(\bar{\mfB}) \leq \dim(\mfB)$.

Consider now the tree model $\bar{\G} = 
(T,\omega,(\bar{\cS}_e,\, e \in E), (\bar{C}_v,\, v \in V))$,
where $\bar{\cS}_e = {\bar{\mfB}|}_e$ for each $e \in E$, and 
$\bar{C}_v = {\bar{\mfB}|}_v$ for each $v \in V$. Note that
$\bar{\cS}_{\he} = \cS_{\he}/W$, and for $e \neq \he$, we have
$\bar{\cS}_e = {\bar{\mfB}|}_e = {\mfB|}_e = \cS_e$. 
All states in $\cS_{\he}$ belonging to the same coset of $W$, say, $\s + W$, 
are mapped to (or merged into) the single state $\s+W$ in $\bar{\cS}_{\he}$. 
Further note that if $v$ is not one of the two vertices incident with $\he$,
then $\bar{C}_v = {\bar{\mfB}|}_v = {\mfB|}_v = C_v$. At the two vertices that
are incident with $\he$, the local constraints are appropriately modified 
to take into account the state-merging at edge $\he$. In any case, we have 
\beq
\dim(\bar{\cS}_e) = \dim({\bar{\mfB}|}_e) \leq \dim({\mfB|}_e) = \dim(\cS_e),
\ \text{ for each } e \in E,
\label{dimSe_ineq}
\eeq
and 
\beq
\dim(\bar{C}_v) = \dim({\bar{\mfB}|}_v) \leq \dim({\mfB|}_v) = \dim(C_v),
\ \text{ for each } v \in V.
\label{dimCv_ineq}
\eeq
\mbox{}\\[-6pt]

We claim that $\bar{\G}$ is an essential tree realization of $\cC$. 
To prove this claim, we must show that ${\mfB(\bar{\G})|}_e = 
{\bar{\mfB}|}_e$ \/ for all $e \in E$, and that 
${\mfB(\bar{\G})|}_I = \cC$, where $\mfB(\bar{\G})$ denotes the 
full behavior of $\bar{\G}$. Note that we do \emph{not} claim 
that $\mfB(\bar{\G}) = \bar{\mfB}$; indeed, this may not be true.

It is easy to see that the inclusion $\cC \subseteq {\mfB(\bar{\G})|}_I$
holds. Indeed, since $\bar{\G} = (T,\omega,({\bar{\mfB}|}_e,\, e \in E),
({\bar{\mfB}|}_v,\, v \in V))$, it is evident that any 
$\bar{\b} \in \bar{\mfB}$ satisfies all the local constraints of 
$\bar{\G}$, and hence is in $\mfB(\bar{\G})$. 
Therefore, $\bar{\mfB} \subseteq \mfB(\bar{\G})$,
and in particular, $\cC = {\bar{\mfB}|}_I \subseteq {\mfB(\bar{\G})|}_I$.

The reverse inclusion, ${\mfB(\bar{\G})|}_I \subseteq \cC$, follows 
from part (a) of the lemma below.

\begin{lemma}
Let $\bar{\b}$ be a global configuration in $\mfB(\bar{\G})$. Then,
\begin{itemize}
\item[(a)] ${\bar{\b}|}_I \in \cC$; and
\item[(b)] ${\bar{\b}|}_{\he} = \0$ if and only if\/
${\bar{\b}|}_I \in \cC_J \oplus \cC_{\bar{J}}$.
\end{itemize}
\label{Gbar_lemma1}
\end{lemma}

We defer the proof of the lemma to Appendix~\ref{Gbar_lemmas_app}.
Lemma~\ref{Gbar_lemma1}(a) shows that ${\mfB(\bar{\G})|}_I \subseteq \cC$,
thus proving that $\bar{\G}$ is a tree realization of $\cC$. It remains
to show that $\G$ is essential, \emph{i.e.}, that ${\mfB(\bar{\G})|}_e = 
{\bar{\mfB}|}_e$ \/ for all $e \in E$. This is shown by the following 
simple argument. We have already seen that $\bar{\mfB} \subseteq 
\mfB(\bar{\G})$, and hence, ${\bar{\mfB}|}_e \subseteq {\mfB(\bar{\G})|}_e$
for all $e \in E$. On the other hand, at any $e \in E$,
${\mfB(\bar{\G})|}_e$ is, by definition, a subspace of 
$\bar{\cS}_e = {\bar{\mfB}|}_e$. Hence, $\bar{\G}$ is essential, 
thus proving our original claim, which we record in the following proposition.

\begin{proposition}
The tree model $\bar{\G}$ is an essential tree realization of $\cC$.
\label{Gbar_prop}
\end{proposition}

Let us call the process described above of obtaining $\bar{\G}$ from
$\G$ as the \emph{state-merging process at edge $\he$}. We use the
notation $\bar{\G} = \merge_{\he}(\G)$ to denote this transformation.
Our goal now is to show that, starting from an essential tree 
realization, if we apply the state-merging process at each edge of 
the underlying tree, then we always end up with a minimal realization. 
A proof of this assertion requires the following technical lemma, 
whose proof we also defer to Appendix~\ref{Gbar_lemmas_app}.

\begin{lemma}
Suppose that there exists
$e' \in E - \{\he\}$ such that the full behavior, $\mfB$, of $\G$
satisfies the following property: for $\b \in \mfB$, we have
${\b|}_{e'} = \0$ if and only if\/
${\b|}_I \in \cC_{J(e')} \oplus \cC_{\bar{J}(e')}$.
Then, for any $\ob \in \mfB(\bar{\G})$, we also have 
${\ob|}_{e'} = \0$ if and only if\/
${\ob|}_I \in \cC_{J(e')} \oplus \cC_{\bar{J}(e')}$.
\label{Gbar_lemma2}
\end{lemma}

We are now in a position to prove the main result of this section,
which provides a construction of $\cM(\cC;T,\omega)$ via state merging.

\begin{theorem}
Let $\G$ be a tree realization of $\cC$
that extends the tree decomposition $(T,\omega)$. 
Let $e_1, e_2, \ldots, e_{|E|}$ be a listing of the edges of 
$T$. Set $\G_0 = \ess(\G)$, and for $i = 1,2,\ldots,|E|$, 
set $\G_i = \merge_{e_i}(\G_{i-1})$. Then, $\G_{|E|}$ is 
the minimal tree realization $\cM(\cC;T,\omega)$.
\label{merge_theorem}
\end{theorem}

\begin{proof}
Let $\mfB$ denote the full behavior of $\G$ (and hence, also of $\ess(\G)$),
and for $i = 1,2,\ldots,|E|$, let $\mfB(\G_i)$ denote the full behavior 
of $\G_i$. By Proposition~\ref{Gbar_prop},
each $\G_i$ is an essential tree realization of $\cC$.

By Lemma~\ref{Gbar_lemma1}(b), for any
$\b \in \mfB(\G_i)$, we have ${\b|}_{e_i} = \0$ if and only if
${\b|}_I \in \cC_{J(e_i)} \oplus \cC_{\bar{J}(e_i)}$. Furthermore,
by Lemma~\ref{Gbar_lemma2}, for any $j \geq i$, $\mfB(\G_j)$ satisfies
the following property: 
\begin{quote}
for any $\b \in \mfB(\G_j)$, we have ${\b|}_{e_i} = \0$ if and only if
${\b|}_I \in \cC_{J(e_i)} \oplus \cC_{\bar{J}(e_i)}$. 
\end{quote}
In particular, $\mfB^* \define \mfB(\G_{|E|})$ satisfies the 
following property for $i = 1,2,\ldots,|E|$: 
\begin{quote}
for any $\b \in \mfB^*$, we have ${\b|}_{e_i} = \0$ if and only if
${\b|}_I \in \cC_{J(e_i)} \oplus \cC_{\bar{J}(e_i)}$. 
\end{quote}
Let us call the above property (P).
Property (P) has two important consequences. Firstly,
it implies that if $\b \in \mfB^*$ is such that ${\b|}_I = \0$,
then ${\b|}_{e} = \0$ for all $e \in E$. This means that the
projection $\pi:\ \mfB^* \rightarrow \cC$ defined by
$\pi(\b) = {\b|}_I$ is in fact an isomorphism.

For the second consequence of (P), consider, for any $e \in E$, 
the homomorphism $\beta_e:\ \cC \rightarrow {\mfB^*|}_e$
defined by $\beta_e(\c) = {(\pi^{-1}(\c))|}_e$. This map
is well-defined since $\pi$ is an isomorphism. Property (P) is equivalent
to the assertion that, for any $e \in E$, the kernel of $\beta_e$
is precisely $\cC_{J(e)} \oplus \cC_{\bar{J}(e)}$. Therefore,
${\mfB^*|}_e \cong \cC/(\cC_{J(e)} \oplus \cC_{\bar{J}(e)})$.

Thus, for each $e \in E$, state space ${\mfB^*|}_e$ is isomorphic to 
$\cS_e^*$ defined in (\ref{Se*_def}), and the map $\beta_e$
is the canonical projection map $\s_e^*$ given by (\ref{se*_def}). 
It easily follows that for each $v \in V$, ${\mfB^*|}_v$ is 
isomorphic to $C_v^*$ defined in (\ref{Cv*_def}). Hence, 
$\G_{|E|} = (T,\omega,({\mfB^*|}_e,\ e \in E), ({\mfB^*|}_v,\ v \in V))$
is the minimal realization $\cM(\cC;T,\omega)$.
\end{proof}

Observe that at each step of the procedure outlined in 
Theorem~\ref{merge_theorem}, the dimensions of the state spaces
and the local constraints do not increase. To make this precise,
given tree models $\G' = (T,\omega,(\cS_e',\ e \in E), (C_v',\ v \in V))$
and $\G'' = (T,\omega,(\cS_e'',\ e \in E), (C_v'',\ v \in V))$,
let us say that $\G' \preccurlyeq \G''$ if 
$\dim(\cS_e') \leq \dim(\cS_e'')$ for all $e \in E$,
and $\dim(C_v') \leq \dim(C_v'')$ for all $v \in V$.
Then, for $\G$ and $\G_i$, $i = 0,1,2\ldots,|E|$, as in the statement of 
Theorem~\ref{merge_theorem}, we have by virtue of (\ref{dimSe_ineq}) 
and (\ref{dimCv_ineq}),
$$
\G_{|E|} \preccurlyeq \G_{|E|-1} \preccurlyeq \ldots \preccurlyeq \G_1 
\preccurlyeq \G_0 = \ess(\G) \preccurlyeq \G.
$$
Thus, we have that if $\G$ is any tree realization of $\cC$
that extends the tree decomposition $(T,\omega)$,
then $\cM(\cC;T,\omega) \preccurlyeq \G$. We record this strong property 
of minimal realizations as a corollary to Theorem~\ref{merge_theorem}.

\begin{corollary}
Let $(T,\omega)$ be a tree decomposition of the index set of a code $\cC$, and 
let $\cM(\cC;T,\omega) = (T,\omega,(\cS_e^*,\ e \in E), (C_v^*,\ v \in V))$
be the corresponding minimal tree realization of $\cC$. Then,
for any tree realization, $(T,\omega,(\cS_e,\ e \in E), (C_v,\ v \in V))$,
of $\cC$ that extends $(T,\omega)$, 
we have $\dim(\cS_e^*) \leq \dim(\cS_e)$ for all $e \in E$,
and $\dim(C_v^*) \leq \dim(C_v)$ for all $v \in V$.
\label{dimCv_cor}
\end{corollary}

The procedure outlined in Theorem~\ref{merge_theorem} does not 
translate to an efficient algorithm for the construction of 
$\cM(\cC;T,\omega)$. This is because the state-merging procedure 
that creates $\G_i$ from $\G_{i-1}$ requires knowledge of the full
behavior of $\G_{i-1}$, which may not be easily determined. 
So, as a practical method for constructing $\cM(\cC;T,\omega)$, given
$\cC$ and $(T,\omega)$, we propose a novel construction
that relies upon the code decomposition techniques of the next
section.

\section{Code Decompositions\label{decomp_section}}

In previous work \cite{kashyap}, it was demonstrated that techniques 
from the decomposition theory of matroids \cite{Sey80},\cite{truemper} 
could be put to good use in a coding-theoretic setting. 
The decomposition theory in that work was presented in the 
context of binary linear codes. As we will now show, the basic 
elements of that theory can be easily extended to cover the case
of nonbinary codes as well. The object of this exercise is not just
to create a more general code decomposition theory, but as we will 
see in the next section, this decomposition theory ties in very nicely 
with the theory of tree realizations.

Let $\cC_1$ and $\cC_2$ be linear codes over the 
finite field\footnote{Up to this point, we did not need to specify
the number of elements in the finite field over which we were working,
but from now on, it will be useful for us to do so.}
$\F_q = GF(q)$, defined on the index sets $I_1$ and $I_2$, respectively. 
Let $I_1 \Delta I_2$ denote the symmetric difference, 
$(I_1 \cup I_2) - (I_1 \cap I_2)$, of the index sets.
We will construct a code $\sfS(\cC_1,\cC_2)$ with $I_1 \Delta I_2$ 
as its index set. For $\x = (x_i,\, i \in I_1) \in \cC_1$
and $\y = (y_i,\, i \in I_2) \in \cC_2$, let 
$\x \star \y = (c_i, \, i \in I_1 \cup I_2)$ be defined by
$$
c_i = \left\{
\begin{array}{cl}
x_i & \text{for $i \in I_1 - I_2$} \\
y_i & \text{for $i \in I_2 - I_1$} \\
x_i-y_i & \text{for $i \in I_1 \cap I_2$}.
\end{array}
\right.
$$
Setting $\cC_1 \star \cC_2 = 
\{\x \star \y:\ \x \in \cC_1,\, \y \in \cC_2\}$,
we see that $\cC_1 \star \cC_2$ has $I_1 \cup I_2$ as its index set.
We take $\sfS(\cC_1,\cC_2)$ to be the cross-section 
${(\cC_1 \star\cC_2)}_{I_1 \Delta I_2}$.
Note that when $I_1 \cap I_2 = \emptyset$, we have
$\sfS(\cC_1,\cC_2) = \cC_1 \star \cC_2 = \cC_1 \oplus \cC_2$.

For $i = 1,2$, let $\cC_i^{(p)}$ and $\cC_i^{(s)}$ denote the
projection ${\cC_i|}_{I_1 \cap I_2}$ and the cross-section
${(\cC_i)}_{I_1 \cap I_2}$, respectively. The codes 
$\cC_i^{(p)}$ and $\cC_i^{(s)}$, for $i = 1,2$, 
all have $I_1 \cap I_2$ as their index set.
The dimension of $\sfS(\cC_1,\cC_2)$ can be expressed in terms of the
codes $\cC_i$, $\cC_i^{(p)}$ and $\cC_i^{(s)}$, $i=1,2$, as 
stated in the following lemma.

\begin{proposition}
For codes $\cC_1,\cC_2$, we have
$$
\dim(\sfS(\cC_1,\cC_2)) = \dim(\cC_1) + \dim(\cC_2) 
 - \dim(\cC_1^{(s)} \cap \cC_2^{(s)}) - \dim(\cC_1^{(p)} + \cC_2^{(p)}),
$$
where $\cC_1^{(p)} + \cC_2^{(p)} = \{\x + \y:\, \x \in \cC_1^{(p)},
\y \in \cC_2^{(p)}\}$.
\label{dim_sum_prop}
\end{proposition}
\begin{proof}
For a code $\cC$, and a subset $J$ of its index set, the kernel
of the projection map $\pi: \cC \rightarrow {\cC|}_\oJ$ is isomorphic
to $\cC_J$, and hence, $\dim(\cC_J) = \dim(\cC) - \dim({\cC|}_\oJ)$.
Thus, taking $\cC = \cC_1 \star \cC_2$, and $J = I_1 \Delta I_2$, 
we find that
$$
\dim(\sfS(\cC_1,\cC_2)) = \dim(\cC_1 \star \cC_2) - 
\dim({(\cC_1 \star \cC_2)|}_{I_1 \cap I_2}) 
= \dim(\cC_1 \star \cC_2) - 
\dim(\cC_1^{(p)} + \cC_2^{(p)}),
$$
since ${(\cC_1 \star \cC_2)|}_{I_1 \cap I_2} = \cC_1^{(p)} + \cC_2^{(p)}$.
So, we must show that $\dim(\cC_1 \star \cC_2) = 
\dim(\cC_1) + \dim(\cC_2) - \dim(\cC_1^{(s)} \cap \cC_2^{(s)})$.

Let $\widetilde{\cC}_2$ be a copy of $\cC_2$ defined on an index set
that is disjoint from $I_1$. For each $\y \in \cC_2$,
denote by $\widetilde{\y}$ its copy in $\widetilde{\cC}_2$. Consider
the homomorphism $\phi: \cC_1 \oplus \widetilde{\cC}_2 
\rightarrow \cC_1 \star \cC_2$ defined by 
$\phi(\x,\widetilde{\y}) = \x \star \y$. Note that 
$\x \star \y = \0$ iff ${\x|}_{I_1 - I_2} = {\y|}_{I_2-I_1} = \0$
and ${\x|}_{I_1 \cap I_2} - {\y|}_{I_1 \cap I_2} = \0$. Equivalently,
$\x \star \y = \0$ iff ${\x|}_{I_1 \cap I_2} \in \cC_1^{(s)}$,
${\y|}_{I_1 \cap I_2} \in \cC_2^{(s)}$, and 
${\x|}_{I_1 \cap I_2} = {\y|}_{I_1 \cap I_2}$. It follows
that the kernel of $\phi$ is isomorphic to 
$$
\{\z:\ \z \in \cC_1^{(s)},\, \z \in \cC_2^{(s)}\}.
$$
which is simply $\cC_1^{(s)} \cap \cC_2^{(s)}$.

Hence, $\dim(\cC_1 \star \cC_2) = \dim(\cC_1 \oplus \widetilde{\cC}_2)
- \dim(\ker(\phi)) = \dim(\cC_1) + \dim(\cC_2) - 
\dim(\cC_1^{(s)} \cap \cC_2^{(s)})$, as desired.
\end{proof}

We will restrict our attention to a particular instance of the 
$\sfS(\cC_1,\cC_2)$ construction, in which we require that the
codes $\cC_i^{(p)}$ and $\cC_i^{(s)}$, $i = 1,2$, take on a specific
form. We need to introduce some notation first.
For each positive integer $r$, set $m_r = (q^r-1)/(q-1)$, and
fix an $r \times m_r$ matrix, which we denote by $D_r$, over $\F_q$,
with the property that each pair of columns of $D_r$ is linearly
independent over $\F_q$. Note that $D_r$ is a parity-check matrix
for an $[m_r,m_r-r]$ Hamming code over $\F_q$ 
(cf.\ \cite[\S~3.3]{vanlint}). Let $\D_r$ denote the dual of this 
Hamming code, \emph{i.e.}, $\D_r$ is the $[m_r,r]$ code over 
$\F_q$ \emph{generated} by $D_r$. The code $\D_r$ is sometimes
referred to as a \emph{simplex code}.

We take a moment to record an important property of the matrix $D_r$
that we will use later. The column vectors of $D_r$ form a maximal
subset of $\F_q^r$ with the property that each pair of vectors 
from the subset is linearly independent over $\F_q$. This is due
to the fact that the number of distinct one-dimensional subspaces
of $\F_q^r$ is precisely $m_r$. Therefore, any (column) vector 
in $\F_q^r$ is a scalar multiple of some column of $D_r$.

Given an $r > 0$, suppose that the codes $\cC_1$ and $\cC_2$, 
defined on the index sets $I_1$ and $I_2$, respectively, are
such that $|I_1 \cap I_2| = m_r$, and for $i = 1,2$, we have
$\cC_i^{(p)} = \D_r$ and $\cC_i^{(s)} = \{\0\}$. In such a case,
$\sfS(\cC_1,\cC_2)$ is called the \emph{$r$-sum} of $\cC_1$ and $\cC_2$,
and is denoted by $\cC_1 \oplus_r \cC_2$. It is convenient to 
extend this definition to the case of $r=0$ as well: when
$|I_1 \cap I_2| = 0$, the $0$-sum $\cC_1 \oplus_0 \cC_2$ is defined 
to be the direct sum $\cC_1 \oplus \cC_2$. 

\begin{example}
Consider the case of codes defined over the binary field $\F_2$.
Note that $\D_1 = \{0,1\}$. Suppose that $|I_1 \cap I_2| = 1$, 
and that the coordinates of $\cC_1$ are $\cC_2$
are ordered so that the index common to $I_1$ and $I_2$ 
corresponds to the last coordinate of $\cC_1$ and the first 
coordinate of $\cC_2$. The conditions necessary for the 1-sum
$\cC_1 \oplus_1 \cC_2$ to be defined can then be stated as
\begin{itemize}
\item[(P1)] $0 \ldots 01$ is not a codeword of $\cC_1$, and the
last coordinate of $\cC_1$ is not identically zero;
\item[(P2)] $10 \ldots 0$ is not a codeword of $\cC_2$, and the
first coordinate of $\cC_2$ is not identically zero.
\end{itemize}
The composite code $\sfS(\cC_1,\cC_2)$ resulting from $\cC_1, \cC_2$
that satisfy (P1), (P2) above was studied in \cite{kashyap}, 
where it was actually called a ``2-sum''.

We would also like to point out that the specialization of our 
$r$-sum operation to the case $r=2$ was called ``$\bar{3}$-sum'' 
in \cite{kashyap}.\footnote{The 2-sum and $\bar{3}$-sum operations 
defined in \cite{kashyap} imposed additional conditions on the 
lengths of the codes involved in the sum, which we have dropped here.}
To add to the confusion, there was in fact an 
operation called ``3-sum'' defined in \cite{kashyap}, but that, 
in a certain sense, dualizes the 2-sum operation we have given 
in this paper.
\end{example}

For $r > 0$, note that if $\cC_i^{(p)}$ and $\cC_i^{(s)}$ ($i = 1,2$)
are in the form needed to define an $r$-sum, then 
$\cC_1^{(p)} + \cC_2^{(p)} = \D_r$, and 
$\cC_1^{(s)} \cap \cC_2^{(s)} = \{\0\}$.
Therefore, as a corollary to Proposition~\ref{dim_sum_prop}, 
we have the following result (which also applies trivially to the 
$r=0$ case).

\begin{corollary}
For $r \geq 0$, if $\cC_1$, $\cC_2$ are such that 
$\cC_1 \oplus_r \cC_2$ can be defined, then 
$$
\dim(\cC_1 \oplus_r \cC_2) = \dim(\cC_1) + \dim(\cC_2) - r.
$$
\label{dim_rsum_cor}
\end{corollary}

An elementary property of direct sums (\emph{i.e.}, 0-sums) is that 
a code $\cC$ is expressible as a direct sum of smaller codes
if and only if there exists a partition $(J,\oJ)$ of the index set 
of $\cC$ such that $\dim({\cC|}_J) + \dim({\cC|}_\oJ) - \dim(\cC) = 0$.
This property extends beautifully to $r$-sums in general. 

\begin{theorem}
Let $\cC$ be a linear code over $\F_q$, defined on the index set $I$,
and let $r$ be a positive integer. Then, the following statements
are equivalent.
\begin{itemize}
\item[(a)] $\cC = \cC_1 \oplus_r \cC_2$ for some codes $\cC_1$, $\cC_2$.
\item[(b)] There exists a partition $(J,\oJ)$ of $I$, with 
$\min\{|J|,|\oJ|\} \geq r$, such that 
$$
\dim({\cC|}_J) + \dim({\cC|}_\oJ) - \dim(\cC) = r.
$$
\end{itemize}
\label{rsum_theorem}
\end{theorem}
\begin{proof}
\underline{(a) $\Rightarrow$ (b)}: See Appendix~\ref{rsum_app}. \\

\noindent \underline{(b) $\Rightarrow$ (a)}: We give here a complete proof 
of this direction of the theorem, as it gives an explicit construction of 
codes $\cC_1$, $\cC_2$ such that $\cC = \cC_1 \oplus_r \cC_2$, given
a partition $(J,\oJ)$ as in (b). The proof generalizes ideas from 
similar constructions presented in \cite{kashyap}.

Let $(J,\oJ)$ be a partition of $I$ such that 
$\dim({\cC|}_J)+\dim({\cC|}_\oJ)-\dim(\cC) = r$. 
Set $n = |I|$ and $k = \dim(\cC)$, and let $G$ be a $k \times n$ 
generator matrix for $\cC$. Without loss of generality, 
we may assume that the columns of $G$ are ordered so 
that the first $|J|$ columns are indexed by the elements of $J$, 
and the rest by the elements of $\oJ$. In the following exposition,
we will often permute the columns of $G$ to bring the matrix into some 
desired form. Whenever this is the case, it will be tacitly assumed
that column indices migrate with the columns.

Let ${G|}_J$ and ${G|}_{\oJ}$ denote the restrictions of $G$ to 
the columns indexed by the elements of $J$ and $\oJ$, respectively;
thus, $G = \left[{G|}_J \ \ {G|}_{\oJ}\right]$. Let $\rank({G|}_J) = k_1$ 
and $\rank({G|}_{\oJ}) = k_2$; by our assumption on $(J,\oJ)$, we have
we have $k_1 + k_2 = k + r$. 

Bring $G$ into reduced row-echelon form (rref) over $\F_q$. 
Permuting within the columns of ${G|}_J$ and within those of
${G|}_\oJ$ if necessary, $\rref(G)$ may be assumed to be of the form
\begin{equation}
\overline{G} = 
\left[
\begin{array}{cccc}
I_{k_1} & A & \O & B \\
\O & \O & I_{k-k_1} & C
\end{array}
\right],
\label{rref_eq}
\end{equation}
where $I_j$, for $j = k_1,k_2-r$, denotes the $j \times j$ identity matrix,
$A$ is a $k_1 \times (|J| - k_1)$ matrix, $B$ is a 
$k_1 \times (|\oJ| - k + k_1)$ matrix, 
$C$ is a $(k-k_1) \times (|\oJ| - k + k_1)$ matrix, 
and the $\O$'s denote all-zeros matrices of appropriate sizes.

The fact that the submatrix 
$
\left[
\begin{array}{cc}
\O & B \\ I_{k-k_1} & C
\end{array}
\right]
$
must have rank equal to $\rank({G|}_{\oJ}) = k_2$ implies that 
$B$ must have rank $k_2 - (k-k_1) = r$. 
Hence, $B$ has $r$ linearly independent rows, 
call them $\b_1, \ldots, \b_r$, which form a basis of 
the row-space of $B$. Permuting the first $k_1$ rows of $\bar{G}$ 
if necessary, we may assume that $\b_1, \ldots, \b_r$ constitute 
the first $r$ rows of $B$. (Permuting these rows of $\bar{G}$ will
also permute the rows of the $I_{k_1}$ matrix, but the effects
of this can be negated by appropriately permuting the first $k_1$ columns 
of $\bar{G}$.) Any row of $B$ is uniquely expressible as a linear
combination (over $\F_q$) of $\b_1,\ldots,\b_r$. In particular, 
for $i=1,2,\ldots,k_1$, the $i$th row of $B$ can be uniquely expressed
as $\sum_{j=1}^r \alpha_{i,j} \b_j$ for some $\alpha_{i,j} \in \F_q$.

Let us denote by $\d_1, \ldots, \d_r$, the rows of the 
$r \times m_r$ generator matrix, $D_r$, of the code $\D_r$. 
Let $X$ be the $k_1 \times m_r$ matrix such that for $i=1,2,\ldots,k_1$,
the $i$th row of $X$ equals $\sum_{j=1}^r \alpha_{i,j} \d_j$, where the
$\alpha_{i,j}$'s are such that the $i$th row of $B$ is 
$\sum_{j=1}^r \alpha_{i,j} \b_j$. Thus, the row-space of $X$ is
the span of $\d_1,\ldots,\d_r$, \emph{i.e.}, it is the code $\D_r$. To the 
columns of $X$, we assign indices from some set $I_X$ disjoint from $I$.

Now, define the $k_1 \times (|J|+m_r)$ matrix
\beq
G_1 = \left[\begin{array}{ccc}
I_{k_1} & A & X \\
\end{array}
\right],
\label{G1_def}
\eeq
allowing the submatrix $[I_{k_1} \ \ A]$ to retain its column indices 
from $\bar{G}$. Also, define the $k \times (|\oJ|+m_r)$ matrix
\beq
G_2 = \left[\begin{array}{ccc} X & \O & B \\ 
\O & I_{k-k_1} & C \end{array}\right],
\label{G2_def}
\eeq
again allowing the submatrix $
\left[
\begin{array}{cc}
\O & B \\ I_{k-k_1} & C
\end{array}
\right]
$
to retain its column indices from $\bar{G}$.
Thus, the index set of the columns of $G_1$ is $I_1 \define J \disj I_X$,
while that of the columns of $G_2$ is $I_2 \define I_X \disj \oJ$.

Finally, for $i = 1,2$, let $\cC_i$ denote the code over $\F_q$
generated by $G_i$. The following facts about $\cC_1$ and $\cC_2$ may
be verified:
\begin{itemize}
\item[(i)] $\dim(\cC_i) = \rank(G_i) = k_i$, $i = 1,2$.
\item[(ii)] $\cC_1 \oplus_r \cC_2$ can be defined, so that by 
Corollary~\ref{dim_rsum_cor}, $\dim(\cC_1 \oplus_r \cC_2) = k_1 + k_2 - r
= k = \dim(\cC)$.
\item[(iii)] All rows of $\bar{G}$ are in $\cC_1 \oplus_r \cC_2$. Since 
$\bar{G}$ generates the same code as $G$ (recall that column indices
get permuted along with columns), we see that $\cC_1 \oplus_r \cC_2$
contains all the codewords of $\cC$.
\end{itemize}
We leave the details of the routine verification of the above facts
to the reader. It only remains to point out that facts (ii) and (iii) 
above show that $\cC_1 \oplus_r \cC_2 = \cC$, thus completing the proof
of the implication (b) $\Rightarrow$ (a).
\end{proof}

The procedure described in the above proof can be formalized into an 
algorithm that takes as input a $k \times n$ generator matrix $G$ 
(over $\F_q$) for $\cC$, and a partition $(J,\oJ)$ of the index set of $\cC$, 
and produces as output generator matrices of two codes $\cC_1$ and $\cC_2$
(and their associated index sets) such that $\cC = \cC_1 \oplus_r \cC_2$,
where $r = \dim({\cC|}_J) + \dim({\cC|}_\oJ) - \dim(\cC)$.
The run-time complexity of this procedure is determined by the following:
\begin{itemize}
\item an rref computation to find $\bar{G}$ as in (\ref{rref_eq});
this can be carried out in $O(k^2n)$ time, 
which is the run-time complexity of bringing a $k \times n$ matrix to 
reduced row-echelon form via elementary row operations;
\item the computations required to identify a basis ($\b_1,\ldots,\b_r$) 
of the row-space of the matrix $B$, and correspondingly the coefficients
$\alpha_{i,j}$; this could be done by computing the rref of $B$, which
would also take $O(k^2n)$ time;
\item the computations needed to determine the $k_1 \times m_r$ matrix $X$; 
each row of the matrix requires $O(rm_r)$ computations, and there are
$k_1 = O(|J|)$ rows, so the computation of $X$ takes $O(|J|rm_r) =
O(|J|rq^r)$ time.
\end{itemize}
Therefore, the entire procedure can be carried out in
$O(k^2n+|J| r q^r)$ time. It is worth noting that the 
run-time complexity of the procedure is polynomial in 
$n$, $k$ and $q$, but exponential in $r$.

\section{A Construction of $\cM(\cC;T,\omega)$ via 
Code Decompositions\label{new_construct_section}}

\begin{figure}
\epsfig{file=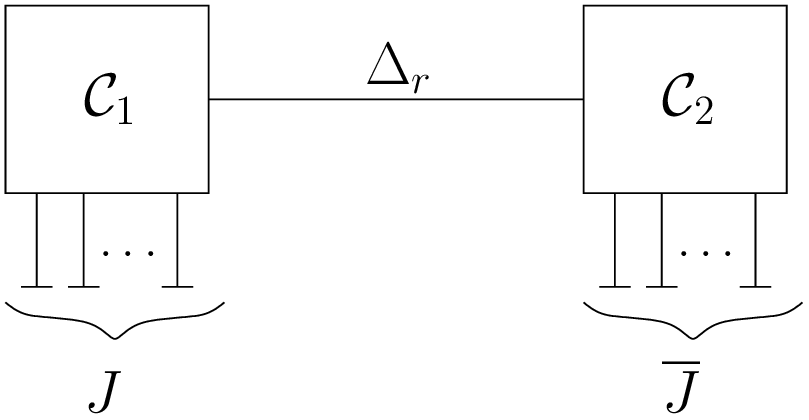, width=5.5cm}
\caption{A tree realization of an $r$-sum decomposition.}
\label{rsum_fig}
\end{figure}

The above procedure for determining an $r$-rum decomposition of a given
code forms the basis of a new construction of minimal tree realizations 
that we present in this section. The key observation behind this
construction is that if a code $\cC$ has a partition $(J,\oJ)$ of
its index set such that $\dim({\cC|}_J) + \dim({\cC|}_\oJ) - \dim(\cC) = r$,
then $\cC$ has an essential tree realization of the form depicted 
in Figure~\ref{rsum_fig}. The tree in the figure consists of a 
single edge $e = \{v_1,v_2\}$, the state space $\cS_e$ is the code $\D_r$,
and the local constraint codes at the two vertices are the codes $\cC_1$
and $\cC_2$ such that $\cC_1 \oplus_r \cC_2 = \cC$. In fact, this
is the minimal realization $\cM(\cC;T,\omega)$,
for the tree $T$ consisting of the single edge $e = \{v_1,v_2\}$, 
and the index map $\omega$ such that $\omega^{-1}(v_1) = J$ and 
$\omega^{-1}(v_2) = \oJ$. This is simply because 
$\dim(\cS_e) = \dim(\D_r) = r$, so by virtue of 
(\ref{dimSe*_alt}), $\cS_e$ has the same dimension as the state space
$\cS^*_e$ in the minimal realization $\cM(\cC;T,\omega)$. So, by 
the uniqueness of minimal tree realizations, the tree realization
depicted in Figure~\ref{rsum_fig} is $\cM(\cC;T,\omega)$.

To summarize, if $\cC$ is a code defined on the index set $I$, 
and $(T,\omega)$ is a tree decomposition of $I$ such that $T$ consists
of the single edge $e=\{v_1,v_2\}$, then we may construct $\cM(\cC;T,\omega)$
as follows. Set $J = \omega^{-1}(v_1)$ and $\oJ = \omega^{-1}(v_2)$,
and compute $r = \dim({\cC|}_J) + \dim({\cC|}_\oJ) - \dim(\cC)$. 
Assign an index set $I_\D$ that is disjoint from $I$ to the code $\D_r$. 
Use the procedure in the proof of Theorem~\ref{rsum_theorem} to
determine codes $\cC_1$ and $\cC_2$, defined on the respective index sets 
$I_1 = J \disj I_\D$ and $I_2 = \oJ \disj I_\D$, such that 
$\cC = \cC_1 \oplus_r \cC_2$. For $i=1,2$,\, assign $\cC_i$ to be the local 
constraint code at vertex $v_i$, and assign $\D_r$ to be the state space
at edge $e$. The resulting tree model $(T,\omega,\D_r,C_1,C_2)$ is
the minimal tree realization $\cM(\cC;T,\omega)$.

Before describing how the construction may be extended to the case of
trees with more than one edge, we deal with the trivial case of 
trees without any edges. If $T$ is a tree consisting of a single 
vertex $v$, and no edges, then given any code $\cC$ defined on some 
index set $I$, there is only one way of realizing $\cC$ on $T$. 
This is the realization $(T,\omega,C_v)$, where $\omega$ is
the unique mapping $\omega: I \rightarrow \{v\}$, and $C_v$ is the
code $\cC$ itself. Of course, this is also the minimal realization 
$\cM(\cC;T,\omega)$.

\begin{figure}
\epsfig{file=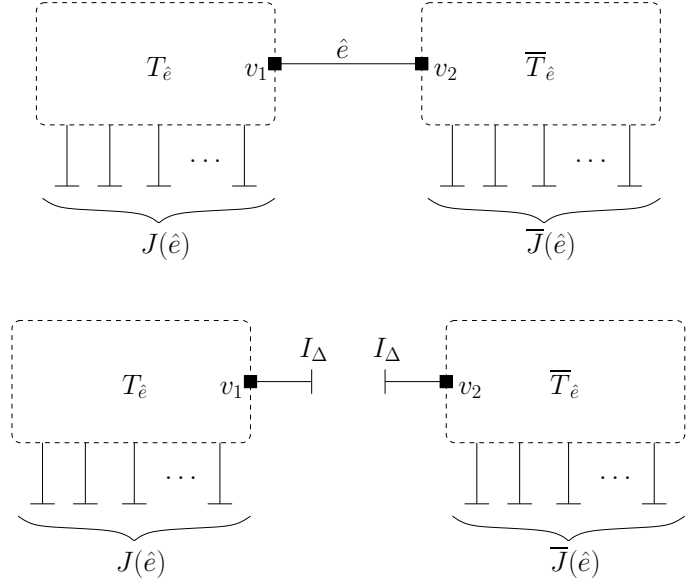, width=9cm}
\caption{Depiction of the manner in which the tree decompositions 
$(T_\he,\omega_1)$ and $(\oT_\he,\omega_2)$ are obtained from $(T,\omega)$.}
\label{rec_constr_fig}
\end{figure}

At this point, we know how to construct $\cM(\cC;T,\omega)$, for any
code $\cC$, and any tree decomposition $(T,\omega)$ such that $T$ has at 
most one edge. From this, we can recursively construct $\cM(\cC;T,\omega)$ 
for any $\cC$ and any $(T,\omega)$, as we now describe.

Suppose that we know how to construct $\cM(\cC;T,\omega)$ for any $\cC$, 
and any $(T,\omega)$ such that $T$ has at most $\eta-1$ edges, for some
integer $\eta \geq 2$. Let $\cC$ be a code defined on the index set $I$,
and let $(T,\omega)$ be a tree decomposition such that $|E(T)| = \eta$.
Pick any $\he = \{v_1,v_2\} \in E(T)$, and as usual, 
let $T_\he$ and $\oT_\he$ be the two components of $T-\he$. 
We will assume that $v_1 \in V(T_\he)$ and $v_2 \in V(\oT_\he)$. 
Let $J(\he) = \omega^{-1}(V(T_\he))$ and 
$\oJ(\he) = \omega^{-1}(V(\oT_\he))$. Compute 
\beq
r = \dim({\cC|}_{J(\he)}) + \dim({\cC|}_{\oJ(\he)}) - \dim(\cC),
\label{r_def}
\eeq
which determines the code $\D_r$. Assign $\D_r$ an index set 
$I_\D$ that is disjoint from $I$. Use the procedure in the proof of 
Theorem~\ref{rsum_theorem} to determine codes $\cC_1$ and $\cC_2$, 
defined on the respective index sets $I_1 = J(\he) \disj I_\D$ and 
$I_2 = \oJ(\he) \disj I_\D$, such that $\cC = \cC_1 \oplus_{r} \cC_2$. 

Now, define the index maps $\omega_1: I_1 \rightarrow V(T_\he)$
and $\omega_2: I_2 \rightarrow V(\oT_\he)$ as follows 
(see Figure~\ref{rec_constr_fig}):
\begin{eqnarray}
\omega_1(i) &=&
\begin{cases}
\omega(i), & \text{ if $i \in J(\he)$} \\
v_1, & \text{ if $i \in I_\D$}
\end{cases} 
\label{omega1_def}
\\
\omega_2(i) &=&
\begin{cases}
\omega(i), & \text{ if $i \in \oJ(\he)$} \\
v_2, & \text{ if $i \in I_\D$}
\end{cases} 
\label{omega2_def}
\end{eqnarray}
Thus, $(T_\he,\omega_1)$ and $(\oT_\he,\omega_2)$ are tree decompositions
of the index sets of $\cC_1$ and $\cC_2$, respectively. As neither
$E(T_\he)$ nor $E(\oT_\he)$ contains the edge $\he$, we have 
$|E(T_\he)| \leq \eta-1$ and $|E(\oT_\he)| \leq \eta-1$. Therefore, by
our assumption, we know how to construct $\cM(\cC_1;T_\he,\omega_1)$ 
and $\cM(\cC_2;\oT_\he,\omega_2)$. Let
\begin{eqnarray}
\cM(\cC_1;T_\he,\omega_1) &=& 
\left(T_\he, \omega_1,\ (\cS_e^{(1)},\, e \in E(T_\he)),\
(C_v^{(1)},\, v \in V(T_\he))\right), \label{minC1_form}
\\
\cM(\cC_2;\oT_\he,\omega_2) &=& 
\left(\oT_\he, \omega_2,\ (\cS_e^{(2)},\, e \in E(\oT_\he)),\
(C_v^{(2)},\, v \in V(\oT_\he))\right). \label{minC2_form}
\end{eqnarray}

Finally, set 
$\G^* = \left(T,\omega,\,(\cS_e,\, e \in E(T)),\,(C_v,\, v \in V(T))\right)$, 
where
\beq
\cS_e = 
\begin{cases}
\cS_e^{(1)}, & \text{ if } e \in E(T_\he) \\
\D_{r}, & \text{ if } e = \he \\
\cS_e^{(2)}, & \text{ if } e \in E(\oT_\he),
\end{cases}
\label{Se_def}
\eeq
and
\beq
C_v =
\begin{cases}
C_v^{(1)}, & \text{ if } v \in V(T_\he) \\
C_v^{(2)}, & \text{ if } v \in V(\oT_\he).
\label{Cv_def}
\end{cases}
\eeq
Figure~\ref{gamma_star} contains a depiction of $\G^*$. 
\begin{figure}
\epsfig{file=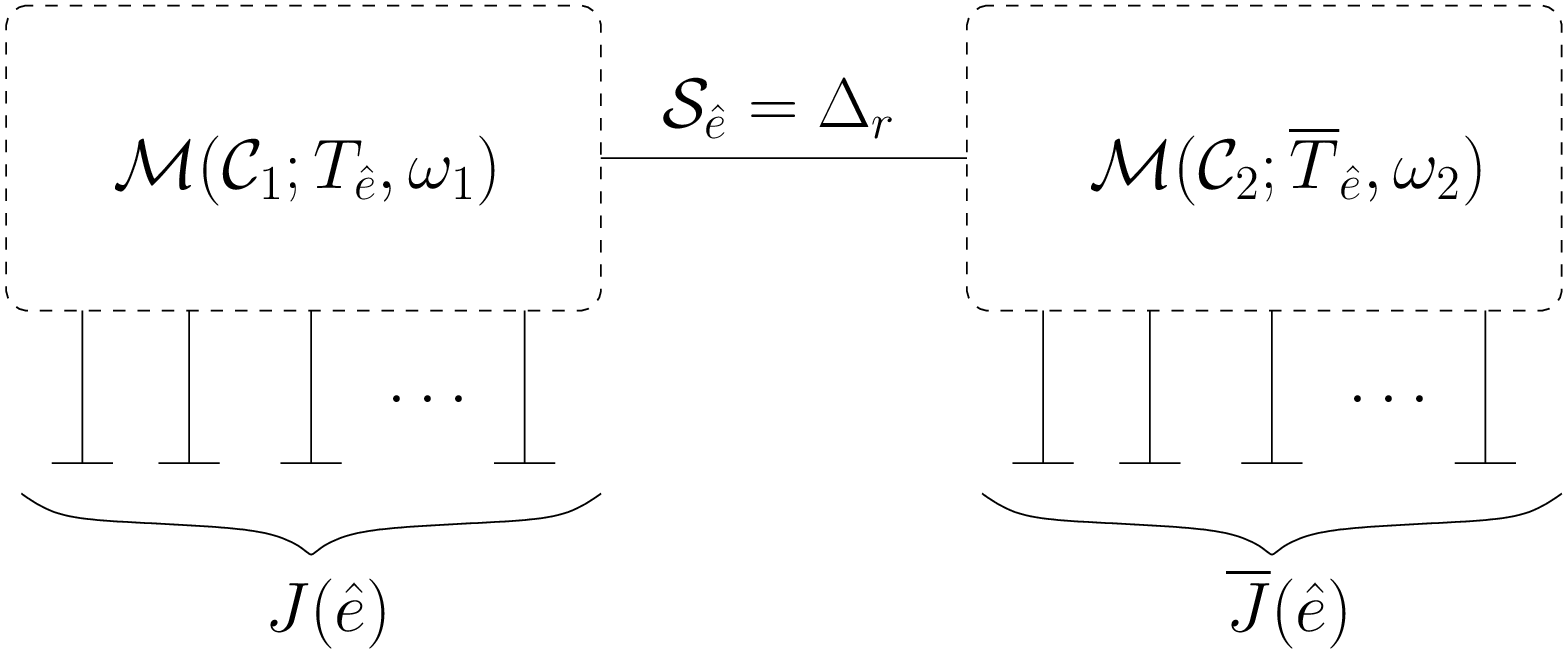, width=8.5cm}
\caption{A depiction of the construction of $\G^*$ from 
$\cM(\cC_1;T_\he;\omega_1)$ and $\cM(\cC_2;\oT_\he,\omega_2)$.}
\label{gamma_star}
\end{figure}
It is easy to see that $\G^*$ is a tree realization of $\cC$. Indeed, 
$\cM(\cC_1;T_\he,\omega_1)$ is a realization of $\cC_1$, and 
$\cM(\cC_2;\oT_\he,\omega_2)$ is a realization of $\cC_2$, 
and hence (as should be clear from Figure~\ref{gamma_star}), 
$\G^*$ is a realization of $\cC_1 \oplus_r \cC_2 = \cC$. 
It is not immediately obvious that $\G^*$ is actually $\cM(\cC;T,\omega)$, 
but this is in fact true, as stated in the following proposition, 
a proof of which is given in Appendix~\ref{rec_constr_app}.

\begin{proposition}
$\G^*$ is the minimal tree realization $\cM(\cC;T,\omega)$.
\label{rec_constr_prop}
\end{proposition}

In summary, we have the following recursive procedure for constructing
$\cM(\cC;T,\omega)$, given a code $\cC$ and a tree decomposition 
$(T,\omega)$. \\[-6pt]

\begin{quote}
\textbf{Procedure} \texttt{MIN\_REALZN}($\cC,T,\omega$) \\[4pt]
\emph{Input\/}: A $k \times n$ generator matrix for a code $\cC$,
and a tree decomposition $(T,\omega)$ of the index set of $\cC$. \\[4pt]
\emph{Output\/}: A specification of the state spaces and the local
constraints in the minimal realization $\cM(\cC;T,\omega)$. \\[4pt]
\textbf{Step M1}. If $T$ consists of a single vertex, then return
$\cM(\cC;T,\omega) = (T,\omega,\cC)$.\\[4pt]
\textbf{Step M2}. 
If $T$ contains at least one edge, then choose an $\he \in E(T)$.
Let $v_1$ be the vertex of $T_\he$ incident with $\he$, and let 
$v_2$ be the vertex of $\oT_\he$ incident with $\he$.
\\
(M2.1)\ Compute
$r = \dim({\cC|}_{J(\he)}) + \dim({\cC|}_{\oJ(\he)}) - \dim(\cC)$. \\
(M2.2)\ Determine $\D_r$, and assign it an index set $I_\D$ disjoint from 
$J(\he) \cup \oJ(\he)$. \\
(M2.3)\ Determine codes $\cC_1$ and $\cC_2$, 
with index sets $I_1 = J(\he) \disj I_\D$ and 
$I_2 = \oJ(\he) \disj I_\D$, respectively,
such that $\cC = \cC_1 \oplus_{r} \cC_2$. \\
(M2.4)\ Determine the index maps $\omega_1$ and $\omega_2$ as in 
(\ref{omega1_def}) and (\ref{omega2_def}). \\[4pt]
\textbf{Step M3}. 
Determine $\cM(\cC_1;T_\he,\omega_1)$ by calling 
\texttt{MIN\_REALZN}($\cC_1,T_\he,\omega_1$);
determine $\cM(\cC_2;\oT_\he,\omega_2)$ by calling 
\texttt{MIN\_REALZN}($\cC_2,\oT_\he,\omega_2$). We may assume that 
$\cM(\cC_1;T_\he,\omega_1)$ and $\cM(\cC_2;\oT_\he,\omega_2)$ are
in the form given in (\ref{minC1_form}) and (\ref{minC2_form}). \\[4pt]
\textbf{Step M4}. Return $\cM(\cC;T,\omega) = 
(T,\omega,\,(\cS_e,\, e \in E(T)),\,(C_v,\, v \in V(T)))$,
where $\cS_e$ and $C_v$ are as defined in (\ref{Se_def}) and (\ref{Cv_def}).
\end{quote}
\mbox{} 

A simplified version of the above procedure may
be obtained by choosing, in Step M2, the edge $\he$ to be an edge incident
with a leaf of $T$. Then, one of the two components of $T - \he$,
say, $\oT_\he$, consists of a single vertex, so that the call
to \texttt{MIN\_REALZN}($\cC_2;\oT_\he,\omega_2$) may be avoided, as it
would simply return $(\oT_\he,\omega_2,\cC_2)$. We will use this 
modification of the procedure to give an estimate of its run-time
complexity. 

Let $n$ denote the length of $\cC$, let $k = \dim(\cC)$, and let $E = E(T)$.
Also, define 
\beq
r_{\max} = 
\max_{e \in E} \ \dim({\cC|}_{J(e)}) + \dim({\cC|}_{\oJ(e)}) - \dim(\cC).
\label{rmax_eq}
\eeq
Observe that, as a result of the modification suggested above, in the
determination of $\cM(\cC;T,\omega)$, the procedure \texttt{MIN\_REALZN}
gets called $|E|$ times, once for each edge $e \in E$. The run-time
complexity of any particular run of \texttt{MIN\_REALZN} is determined
by the computations in Step M2. In the $i$th run,
the procedure acts upon some code $\cC^{(i)}$ of length $n_i$ and 
dimension $k_i$, and in Step M2, it computes an $r_i$,
a code $\D_{r_i}$ with index set $I_\D^{(i)}$, and a code
$\cC_1^{(i)}$. Via Lemma~\ref{appD_lemma}, we have that 
$r_i \leq r_{\max}$. We bound $k_i$ and $n_i$ as follows. 
Note that $\cC_1^{(i)}$ is the code $\cC^{(i+1)}$ that the $(i+1)$th 
run of the procedure takes as input. Thus, we have $k_{i+1} \leq k_i$,
and $n_{i+1} \leq n_i + |I_\D^{(i)}|$. Since $k_1 = k$, $n_1 = n$,
and $|I_\D^{(i)}| = (q^{r_i}-1)/(q-1) \leq q^{r_{\max}}$,
we have, for $i=1,2,\ldots,|E|$, 
$k_i \leq k$ and $n_i \leq n+(i-1) q^{r_{\max}}$. 
Now, by the estimate given in Section~\ref{decomp_section} 
of the run-time complexity of the $r$-sum decomposition procedure, 
we see that the $i$th run of Step M2 of \texttt{MIN\_REALZN} takes 
$O(k_i^2n_i + n_ir_iq^{r_i})$ time. Hence the overall run-time 
complexity of \texttt{MIN\_REALZN} may be estimated to be 
$\sum_{i=1}^{|E|} O(k_i^2n_i + n_ir_iq^{r_i})$.
This expression can be simplified by observing that
\begin{eqnarray*}
\sum_{i=1}^{|E|} (k_i^2n_i + n_ir_iq^{r_i}) 
&\leq& (k^2 + r_{\max}q^{r_{\max}}) \sum_{i=1}^{|E|} n_i \\
&\leq& (k^2 + r_{\max}q^{r_{\max}}) 
\sum_{i=1}^{|E|} \left[n+(i-1) q^{r_{\max}}\right] \\
&=& (k^2 + r_{\max}q^{r_{\max}}) \, 
\left[n|E| + (1/2)|E|(|E|-1)q^{r_{\max}}\right].
\end{eqnarray*}
It follows that \texttt{MIN\_REALZN} runs in
$O((k^2 + r_{\max}q^{r_{\max}})(n|E| + |E|^2 r_{\max}q^{r_{\max}}))$
time. Note that this is polynomial in $n$, $k$, $q$ and $|E|$, but 
exponential in $r_{\max}$.

\section{Complexity Measures\label{complexity_section}}

\subsection{Complexity Measures for Codes\label{code_complexity_section}}
As observed in \cite{For01}, any graphical realization of a code 
specifies an associated decoding algorithm, namely, the sum-product
algorithm. The sum-product algorithm specified by a tree realization,
$\G = (T,\omega,(\cS_e,\, e \in E),\,(C_v,\,v \in V))$, of a code $\cC$
provides an exact implementation of ML decoding for $\cC$. 
A reasonable initial estimate of the computational complexity of the 
sum-product algorithm on $\G$ is provided by the \emph{constraint complexity}
of $\G$, which is defined as $\max_{v \in V} \dim(C_v)$.
As implied by Corollary~\ref{dimCv_cor}, given a tree decomposition 
$(T,\omega)$ of the index set of $\cC$, the minimal realization
$\cM(\cC;T,\omega)$ has the least constraint complexity among all
tree realizations of $\cC$ that extend $(T,\omega)$. Let $\k(\cC;T,\omega)$
denote the constraint complexity of $\cM(\cC;T,\omega)$. Note that,
by (\ref{dimCv*}), 
\beq
\k(\cC;T,\omega) = \max_{v \in V} 
\left(\dim(\cC) - \sum_{e \in E(v)} \dim(\cC_{J(e)})\right).
\label{k_CT_omega}
\eeq
Thus, $\k(\cC;T,\omega)$ is a measure of the complexity of implementing
ML decoding for $\cC$ as a sum-product algorithm on $\cM(\cC;T,\omega)$.

Let us now define the \emph{treewidth} of the code $\cC$ to be
\beq
\k(\cC) = \min_{(T,\omega)} \k(\cC;T,\omega),
\label{treewidth_def_0}
\eeq
where the minimum is taken over all tree decompositions $(T,\omega)$
of the index set of $\cC$. The treewidth of a code is an indicator of how
small the computational complexity of an ML decoding algorithm for $\cC$
can be. The notion of treewidth (\emph{i.e.,} minimal constraint complexity)
of a code was first introduced by Forney \cite{For03}. 
A related notion, called minimal tree complexity, was defined and
studied by Halford and Chugg \cite{halford}. Treewidth, as defined in
(\ref{treewidth_def_0}), is an upper bound on the minimal tree complexity
measure of Halford and Chugg.

A tree is called \emph{cubic} if all its internal nodes have degree 3.
Forney \cite{For03} showed that the minimum in (\ref{treewidth_def_0})
is always achieved by a tree decomposition $(T,\omega)$ in which 
$T$ is a cubic tree, and $\omega$ is a bijection\footnote{Forney 
\cite{For03} only explicitly states that the minimizing $(T,\omega)$ 
may be taken to be such that $T$ is a cubic tree and $\omega$ is
a surjective map onto the leaves of $T$. However, the symbol-splitting
argument in Section~V.F of his paper actually implies that 
$\omega$ in the minimizing tree decomposition may be taken to be 
one-to-one as well.}
between the index set of $\cC$ and the set of leaves of $T$. 
Let $\cQ(\cC)$ denote the set of all tree decompositions $(T,\omega)$
in which $T$ is cubic and $\omega$ maps the index set
of $\cC$ bijectively onto the set of leaves of $T$. We may then re-write
(\ref{treewidth_def_0}) as
\beq
\k(\cC) = \min_{(T,\omega) \in \cQ(\cC)} \k(\cC;T,\omega).
\label{treewidth_def}
\eeq

An alternate measure of code complexity may be obtained 
from the notion of \emph{state complexity} of a tree realization $\G$, 
which is the largest dimension of a state space in $\G$. Thus, 
by virtue of (\ref{dimSe*}) and (\ref{dimSe*_alt}), the 
state complexity of a minimal realization $\cM(\cC;T,\omega)$ is given by
\begin{eqnarray}
\sg(\cC;T,\omega) & = &
\max_{e \in E} \ \dim(\cC) - \dim(\cC_{J(e)}) - \dim(\cC_{\bar{J}(e)}) 
\notag \\
&=& \max_{e \in E} \
\dim({\cC|}_{J(e)}) + \dim({\cC|}_{\bar{J}(e)}) - \dim(\cC).
\label{s_CT_omega}
\end{eqnarray}
We then define, in analogy with (\ref{treewidth_def}),
\beq
\sg(\cC) = \min_{(T,\omega) \in \cQ(\cC)} \sg(\cC;T,\omega).
\label{branchwidth_def}
\eeq
Note that the minimum in the above definition is taken over tree 
decompositions in $\cQ(\cC)$ only. It must be emphasized that
$\sg(\cC)$, as defined in (\ref{branchwidth_def}), need \emph{not} 
be the same as the least $\sg(\cC;T,\omega)$ over all tree decompositions
$(T,\omega)$ of the index set of $\cC$. 

A notion analogous to $\sg(\cC)$ is known as branchwidth in the matroid
theory literature; see \emph{e.g.}, \cite{HOSG07}. 
In keeping with that nomenclature, we will call $\sg(\cC)$ 
the \emph{branchwidth} of the code $\cC$. Branchwidth and treewidth
are very closely related, as shown by the following result,
which can be obtained in a straightforward manner from the bounds in 
Lemma~\ref{dimCv*_bnd}.

\begin{proposition}[\cite{HW06}, Theorem~4.2]
Given a code $\cC$, if $(T,\omega) \in \cQ(\cC)$, then 
$$
\sg(\cC;T,\omega) \leq \k(\cC;T,\omega) \leq 2\sg(\cC;T,\omega).
$$
Hence, $\sg(\cC) \leq \k(\cC) \leq 2\sg(\cC)$.
\label{bw_tw_bnds}
\end{proposition}

The notions of state and constraint complexity have been studied
extensively in the context of conventional trellis realizations of a
code; see \emph{e.g.}, \cite{vardy}. Recall from 
Example~\ref{trellis_example} that a conventional trellis realization
of a code is a tree realization that extends a tree decomposition 
$(T,\omega)$ in which $T$ is a simple path and $\omega$ is a bijection
between the index set of $\cC$ and the vertices of $T$. This special
case of a tree decomposition is referred to as a \emph{path decomposition}.
Specifically, a path decomposition of a code $\cC$ defined on the 
index set $I$ is a pair $(T,\omega)$, where $T$ is a simple path on 
$|I|$ vertices, and $\omega:I \rightarrow V(T)$ is a bijection. Let 
$\cP(\cC)$ denote the set of all path decompositions of $\cC$. We then
define
\beq
\k_\tr(\cC) = \min_{(T,\omega) \in \cP(\cC)} \k(\cC;T,\omega)
\label{k_trellis_def}
\eeq
and
\beq
\sg_\tr(\cC) = \min_{(T,\omega) \in \cP(\cC)} \sg(\cC;T,\omega).
\label{s_trellis_def}
\eeq
It is well-known, and indeed readily follows from Lemma~\ref{dimCv*_bnd},
that $\sg_\tr(\cC) \leq \k_\tr(\cC) \leq \sg_\tr(\cC) + 1$.

It is clear from (\ref{treewidth_def_0}) and (\ref{k_trellis_def}) that
$\k(\cC) \leq \k_\tr(\cC)$. Forney \cite{For03} asked the question of
whether $\k(\cC)$ could be significantly smaller than $\k_\tr(\cC)$. He 
conjectured that either $\k_\tr(\cC) - \k(\cC) \leq 1$ for all codes $\cC$, 
or $\k_\tr(\cC) - \k(\cC)$ is unbounded. We show here that it is in fact
the latter that is true. To do so, we need to introduce some new concepts.

\subsection{Complexity Measures for Graphs\label{graph_complexity_section}}
In their fundamental work on graph minors \cite{RS-survey}, 
Robertson and Seymour introduced two notions of complexity of 
graphs, namely, treewidth and pathwidth. These notions have proved to
be invaluable tools with many applications in graph theory and 
theoretical computer science. An overview of such applications can
be found, for example, in \cite{bod93}. We will define the notions 
of treewidth and pathwidth of a graph in this subsection, 
and subsequently, relate them to the complexity measures 
$\k(\cC)$ and $\k_\tr(\cC)$ defined above for codes.

Let $\cG$ be a graph with vertex set $V(\cG)$ and edge set $E(\cG)$. 
The graph may contain self-loops and parallel edges. 
A \emph{tree decomposition} of $\cG$ is a pair $(T,\beta)$, where $T$ 
is a tree, and $\beta: V(T) \rightarrow 2^{V(\cG)}$ is a mapping that
satisfies the following:
\begin{itemize}
\item[(T1)] $\bigcup_{x \in V(T)} \beta(x) = V(\cG)$;
\item[(T2)] for each pair of adjacent vertices $u,v \in V(\cG)$, we have
$\{u,v\} \subseteq \beta(x)$ for some $x \in V(T)$; and
\item[(T3)] for each pair of vertices $x,z \in V(T)$, if $y \in V(T)$
is any vertex on the unique path between $x$ and $z$, then 
$\beta(x) \cap \beta(z) \subseteq \beta(y)$.
\end{itemize}
It may be helpful to point out that (T3) above is equivalent to the
following: 
\begin{itemize}
\item[(T3$'$)] for each $v \in V(\cG)$, the subgraph of $T$ induced 
by $\{x \in V(T): v \in \beta(x)\}$ is a (connected) subtree of $T$.
\end{itemize}
A reader familiar with the notion of ``junction trees'' (see \emph{e.g.},
\cite{AM00}) will recognize a tree decomposition of $\cG$ to be a 
junction tree. 

The \emph{width} of a tree decomposition $(T,\beta)$ as above
is defined to be $\max_{x \in V(T)} |\beta(x)|-1$. The \emph{treewidth} 
of $\cG$, which we denote by $\k(\cG)$, is the minimum among the widths of 
all its tree decompositions. Note that if $\cG$ has at least one edge, 
then, because of (T2), any tree decomposition of $\cG$ must have width 
at least one. Thus, for any graph $\cG$ with $|E(\cG)| \geq 1$,
we have $\k(\cG) \geq 1$.

\begin{example}
For any tree $T$ with at least two vertices, we have $\k(T) = 1$. 
This can be seen as follows. Fix a vertex $r \in V(T)$. Define 
a mapping $\beta:V(T) \rightarrow 2^{V(T)}$ as follows: 
$\beta(r) = \{r\}$, and for $x \neq r$, $\beta(x) = \{x,y\}$,
where $\{x,y\}$ is the first edge on the unique path from $x$ to $r$.
It is easily verified that $(T,\beta)$ is a tree decomposition of $T$.
Since this tree decomposition has width one, it follows that $\k(T) = 1$.
\label{treewidth_example1}
\end{example}

If $(T,\beta)$ is a tree decomposition in which $T$ is a simple path,
then $(T,\beta)$ is called a \emph{path decomposition}. The minimum
among the widths of all the path decompositions of $\cG$ is called
the \emph{pathwidth} of $\cG$, which we denote by $\k_\path(\cG)$. 
It is evident that $\k(\cG) \leq \k_\path(\cG)$.

Analogous to the situation of Example~\ref{treewidth_example1}, 
a simple path has pathwidth one. However, trees may have arbitrarily
large pathwidth. The following example is due to Robertson and Seymour
\cite{RS-I}. 

\begin{example}
Let $Y_1$ be the complete bipartite graph $K_{1,3}$. For $i \geq 2$,
we inductively define $Y_i$ by taking a copy of $Y_{i-1}$, and to each
leaf $v$ of this graph, adding two new vertices adjacent to $v$. 
Figure~\ref{Y_trees} shows the trees $Y_1$, $Y_2$ and $Y_3$.
\begin{figure}
\epsfig{file=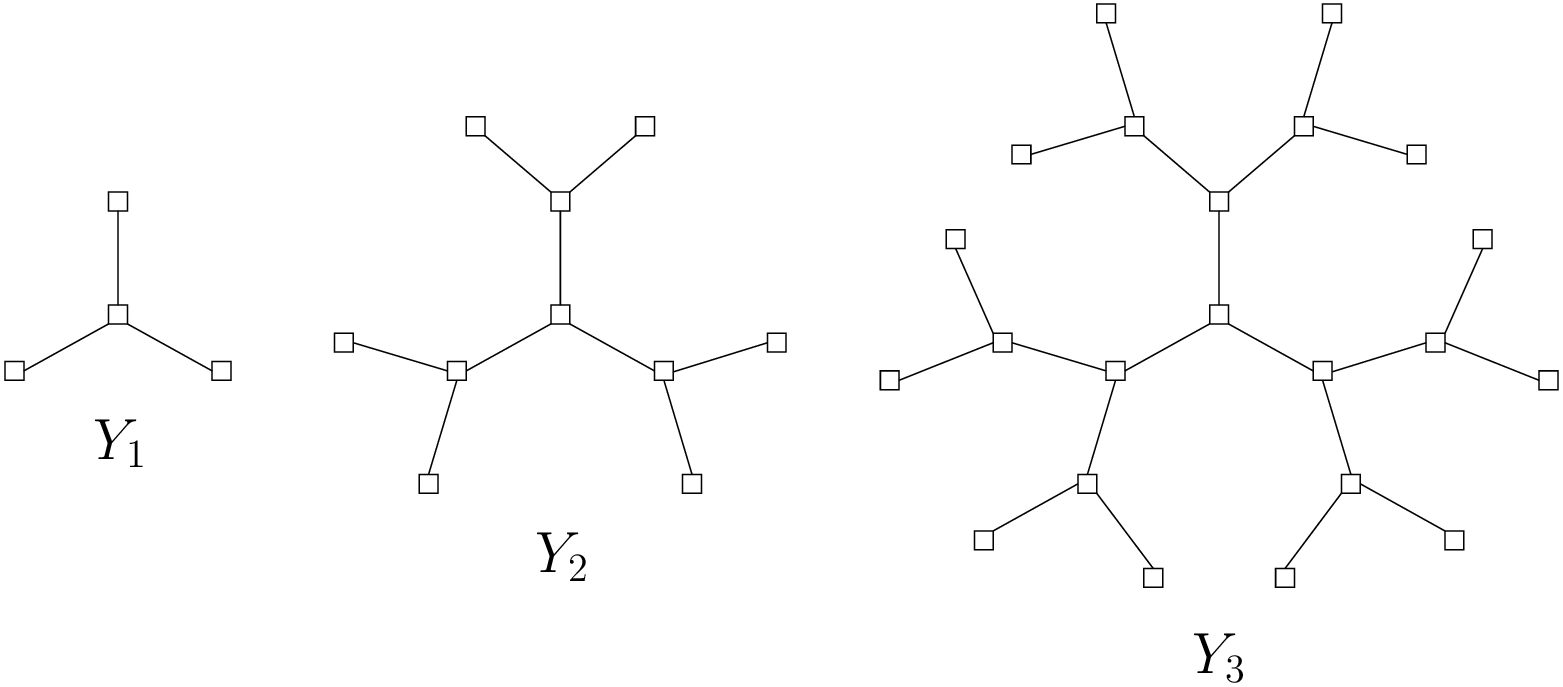, width=9.5cm}
\caption{The trees $Y_1$, $Y_2$ and $Y_3$.}
\label{Y_trees}
\end{figure}
The pathwidth of $Y_i$, $i \geq 1$, is $\lceil \frac{1}{2}(i+1) \rceil$
\cite{RS-I}.
\label{treewidth_example2}
\end{example}

Thus, for trees $T$, the difference $\k_\path(T) - \k(T)$ can be 
arbitrarily large. We will use this fact to construct codes $\cC$ 
for which $\k_\tr(\cC) - \k(\cC)$ is arbitrarily large.

We remark that the problem of determining the treewidth or pathwidth
of a graph is known to be NP-hard \cite{arnborg},\cite{bod93}.
As we will see a little later, this implies that the problem
of determining the treewidth of a code, or its trellis counterpart,
is also NP-hard.

\subsection{Relating the Complexity Measures for Codes and 
Graphs\label{relate_section}}
Let $\F$ be an arbitrary finite field. To any given graph $\cG$,
we will associate a code $\cC[\cG]$ over $\F$ as follows.
Let $D(\cG)$ be any directed graph obtained by arbitrarily 
assigning orientations to the edges of $\cG$,
and let $A_{D(\cG)}$ be the vertex-edge incidence matrix of $D(\cG)$. 
This is the $|V(\cG)| \times |E(\cG)|$ matrix whose rows and columns
are indexed by the vertices and directed edges, respectively, of $D(\cG)$,
and whose $(i,j)$th entry, $a_{i,j}$, is determined as follows:
$$
a_{i,j} = 
\begin{cases}
1 & \text{if vertex $i$ is the tail of non-loop edge $j$} \\
-1 & \text{if vertex $i$ is the head of non-loop edge $j$} \\
0 & \text{otherwise}.
\end{cases}
$$
The code $\cC[\cG]$ is defined to be the linear code over $\F$ generated by 
the matrix $A_{D(\cG)}$. When $\F$ is the binary field, the code $\cC[\cG]$
is the \emph{cut-set code} of $\cG$, \emph{i.e.}, the dual of the cycle 
code of $\cG$ \cite{HB68}.

The following fundamental result that relates the treewidths of the graph 
$\cG$ and the code $\cC[\cG]$ is due to Hlin{\v e}n\'y and 
Whittle\footnote{The results in \cite{HW06} are stated in matroid-theoretic
language. The vocabulary necessary to translate the language of 
matroid theory into that of coding theory can be found, for example,
in \cite{kashyap_siam}.} \cite{HW06}. 

\begin{theorem}[\cite{HW06}, Theorem~3.2]
If $\cG$ is a graph with at least one edge, then $\k(\cG) = \k(\cC[\cG])$.
\label{treewidth_theorem}
\end{theorem}

Since determining the treewidth of a graph is NP-hard, it immediately 
follows from the above theorem that the problem of determining the 
treewidth of a code (over any fixed finite field) is also NP-hard.
We remark that the problem of determining the branchwidth of a code
is also NP-hard. This follows from a result \cite{HM07}
that relates the branchwidth of the code $\cC[\cG]$ 
to the branchwidth of the graph $\cG$, the
latter being a notion we have not defined in this paper.

Unfortunately, it is not true that $\k_\path(\cG) = \k_\tr(\cC[\cG])$.
As an example, consider the code $\cC[T]$ over the binary field, for
an arbitrary tree $T$. It is not hard to see that 
$\cC[T] = \{0,1\}^{|E(T)|}$ which, being the direct sum of multiple
copies of $\{0,1\}$, has $\k_\tr(\cC[T]) = 1$. But as we have already
noted, trees can have arbitrarily large pathwidth.

\begin{figure}[t]
\centering{\epsfig{file=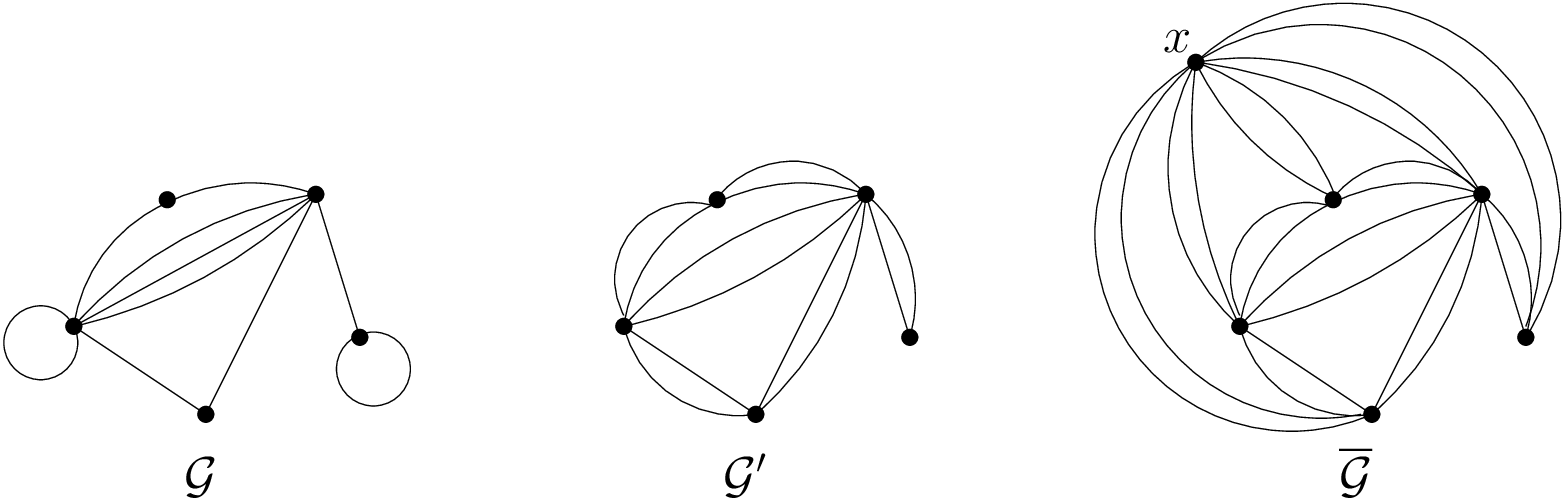, width=11cm}}
\caption{Construction of $\cG'$ and $\bar{\cG}$ from $\cG$.}
\label{Gbar}
\end{figure}

We get around this problem by means of a suitable transformation of graphs.
Given a graph $\cG$, let $\cG'$ be a graph defined on the same vertex set 
as $\cG$, having the following properties (see Figure~\ref{Gbar}):
\begin{itemize}
\item $\cG'$ is loopless;
\item a pair of distinct vertices is adjacent in $\cG'$
iff it is adjacent in $\cG$; and
\item in $\cG'$, there are exactly two edges between 
each pair of adjacent vertices.
\end{itemize}
Define $\bar{\cG}$ to be the graph obtained by adding an extra vertex, 
$x$, to $\cG'$, along with a pair of parallel edges from $x$ to each 
$v \in V(\cG')$ (see Figure~\ref{Gbar}). It is easy to see that 
$\bar{\cG}$ is constructible directly from $\cG$ in $O(|V(\cG)|^2)$ time. 

The following result was used in \cite{kashyap_siam} to show that the problem
of determining $\sg_\tr(\cC)$ for an arbitrary code $\cC$ (over any
fixed finite field) is NP-hard. 

\begin{theorem}[\cite{kashyap_siam}, Proposition~3.1]
If $\bar{\cG}$ is the graph constructed from a given graph $\cG$ as
described above, then $\sg_\tr(\cC[\bar{\cG}]) = \k_\path(\cG) + 1$.
\label{pathwidth_theorem}
\end{theorem}

Since $\sg_\tr(\cC)$ is always within one of $\k_\tr(\cC)$, the above 
theorem implies that 
\beq
\k_\path(\cG) + 1 \leq \k_\tr(\cC[\bar{\cG}]) \leq \k_\path(\cG) + 2.
\label{k_trellis_bnd}
\eeq
While this falls short of establishing the NP-hardness of computing
$\k_\tr(\cC)$ for an arbitrary code $\cC$, it is certainly enough to
provide us with the desired example of codes $\cC$ for which 
$\k_\tr(\cC) - \k(\cC)$ is arbitrarily large. We just need to make
one more observation: $\k(\bar{\cG}) = \k(\cG)+1$. The proof of this
fact, which is along the lines of the proof of Lemma~3.5 
in \cite{kashyap_siam}, is left to the reader as a straightforward
exercise. We can now prove the following corollary to 
Theorems~\ref{treewidth_theorem} and \ref{pathwidth_theorem}.

\begin{corollary}
Over any finite field $\F$, there exists a family of codes 
$\cC_i$, $i \in 1,2,\ldots$, such that 
$$
\lim_i\ \k_\tr(\cC_i) - \k(\cC_i) = \infty.
$$
\label{pathwidth_cor}
\end{corollary}
\begin{proof}
Let $Y_i$, $i = 1,2,\ldots$, be the family of trees defined in 
Example~\ref{treewidth_example2}. Define $\cC_i = \cC[\bar{Y_i}]$,
where $\bar{Y_i}$ refers to the graph obtained from $Y_i$ by the
transformation depicted in Figure~\ref{Gbar}. Note that 
$\k(\bar{Y_i}) = \k(Y_i)+1 = 2$, since $\k(Y_i) = 1$, as shown
in Example~\ref{treewidth_example1}. Thus, on the one hand, from 
Theorem~\ref{treewidth_theorem}, we have $\k(\cC_i) = \k(\bar{Y_i}) = 2$.
And on the other hand, from (\ref{k_trellis_bnd}) and 
Example~\ref{treewidth_example2}, we have 
$\k_\tr(\cC_i) \geq \k_\path(Y_i) + 1 = \lceil \frac{1}{2}(i+3) \rceil$.
\end{proof}

Using standard facts known about the incidence matrix $A_{D(\cG)}$ for
a graph $\cG$ (see, for example, \cite[Chapter~5]{oxley}), it may be
verified that the codes $\cC_i$, $i \geq 1$, constructed in the above proof 
are $[n_i,k_i,d_i]$ codes, where 
\begin{eqnarray*}
n_i &=& |E(\bar{Y_i})| \ \ =\ \ 12(2^i - 1) + 2, \\ 
k_i &=& |V(\bar{Y_i})|-1 \ \ = \ \ 3(2^i-1) + 1, \\
d_i &=& \text{size of the smallest cut-set in $\bar{Y_i}$} \ \ = \ \ 4.
\end{eqnarray*}
Note that $\k_\tr(\cC_i) - \k(\cC_i)$ grows as $O(\log n_i)$. 
We conjecture that this is in fact the maximal rate of growth 
of the difference $\k_\tr(\cC_i) - \k(\cC_i)$ for any code family $\cC_i$.

\begin{conjecture}
If $\cC_i$, $i \geq 1$, is any sequence of codes over $\F$, then
$$
\limsup_i \frac{\k_\tr(\cC_i) - \k(\cC_i)}{\log n_i} < \infty,
$$
where $n_i$ denotes the length of the code $\cC_i$.
\label{conj1}
\end{conjecture}

The codes $\cC_i$ constructed in the proof of 
Corollary~\ref{pathwidth_cor} all have treewidth equal to two. 
Issues related to families of codes whose treewidth is bounded 
by a constant are discussed next.

\subsection{Codes of Bounded Complexity\label{bounded_section}}
Many NP-hard combinatorial problems on graphs are known to be solvable in 
polynomial (often, linear) time when restricted to graphs of bounded 
treewidth \cite{AP89},\cite{bod88}. In this subsection, we will see 
that the same general principle applies to problems pertaining to codes 
as well.

Let $\F_q = GF(q)$ be a fixed finite field. Given an integer $t \geq 0$,
denote by $\TW(t)$ (resp.\ $\BW(t)$) the family of all codes over $\F_q$
of treewidth (resp.\ branchwidth) at most $t$. Thus, a family $\mfC$ of 
codes over $\F_q$ is said to have bounded treewidth (resp.\ branchwidth)
if $\mfC \subseteq \TW(t)$ (resp.\ $\mfC \subseteq \BW(t)$) for some
integer $t$. Note that by Proposition~\ref{bw_tw_bnds},
$\BW(\lfloor t/2 \rfloor) \subseteq \TW(t) \subseteq \BW(t)$, and
so, a code family $\mfC$ has bounded treewidth if and only if it has bounded
branchwidth. 
 
A fundamental result of coding theory \cite{BMvT} states that the problem 
of ML decoding is NP-hard for an arbitrary family of codes. However,
we will now show that this problem becomes solvable in linear time for 
any code family of bounded treewidth. So, consider a code family
$\mfC \subseteq \TW(t)$, where $t$ is a \emph{fixed} integer, and 
pick an arbitrary $\cC \in \mfC$. Let $n$ denote the length of $\cC$. 
By definition, $\cC$ has a minimal realization $\cM(\cC;T,\omega)$ 
with constraint complexity at most $t$. 
Moreover, by (\ref{treewidth_def}), $(T,\omega)$ can be chosen
to be in $\cQ(\cC)$, \emph{i.e.}, it may be chosen so that
$T$ is a cubic tree, and $\omega$ maps the index set of $\cC$
bijectively onto the leaves of $T$. In particular, the number of 
leaves of $T$ equals the cardinality, $n$, of the index set of $\cC$.

Now, recall that ML decoding of $\cC$ may be implemented as a sum-product
algorithm on any tree realization of $\cC$, and in particular, 
on $\cM(\cC;T,\omega)$. The computational complexity of the 
sum-product algorithm on $\cM(\cC;T,\omega)$ is determined
by the computations that take place at the internal nodes of $T$. 
By an estimate of Forney \cite[Theorem~5.2]{For01}, the 
number of computations at the internal node $v \in V(T)$ is 
of the order of $\delta_v (\delta_v-2) q^{\dim(C_v^*)}$, where 
$\delta_v$ is the degree of $v$ in $T$. Since $T$ is cubic, $\delta_v = 3$,
and since the constraint complexity of $\cM(\cC;T,\omega)$ is at most $t$,
we have $\dim(C_v^*) \leq t$. Hence, the number of computations performed
by the sum-product algorithm at any internal node of $T$ is bounded
by $3q^t$, which is a constant. Now, $T$ is a cubic tree on $n$ leaves,
so it has at most $n-2$ internal nodes. It follows that the computational
complexity of the sum-product algorithm on $\cM(\cC;T,\omega)$ is $O(n)$,
the constant in the $O$-notation being proportional to $3q^t$. 
Thus, there is a linear-time implementation of ML decoding for any 
$\cC \in \mfC$.

A question that naturally arises in this context is that of how hard 
it is to explicitly determine the minimal tree realization required for 
linear-time implementation of ML decoding. Note that this is not exactly
a decoding complexity issue, since the determination of a suitable
$\cM(\cC;T,\omega)$ may be done ``off-line'' for each $\cC \in \mfC$.

An explicit determination of $\cM(\cC;T,\omega)$ involves finding a
tree decomposition $(T,\omega) \in \cQ(\cC)$ such that 
$\k(\cC;T,\omega) \leq t$, and the specification of the state spaces
and the local constraint codes of $\cM(\cC;T,\omega)$. Given a 
$(T,\omega) \in \cQ(\cC)$ satisfying $\k(\cC;T,\omega) \leq t$,
the state spaces and local constraint codes of $\cM(\cC;T,\omega)$ may
be determined by the \texttt{MIN\_REALZN} procedure of 
Section~\ref{new_construct_section}. An estimate of the computational
complexity of this procedure was given in that section, in terms
of the length and dimension of $\cC$, the number of edges in $T$,
and the quantity $r_{\max}$ defined in (\ref{rmax_eq}). Comparing 
(\ref{rmax_eq}) with (\ref{s_CT_omega}), we see that $r_{\max}$ 
is simply $\sg(\cC;T,\omega)$, which by Proposition~\ref{bw_tw_bnds},
is bounded from above by $t$. The number of edges of $T$ is $|V(T)|-1$,
and $V(T)$ consists of $n$ leaves and at most $n-2$ internal nodes, $n$
being the length of $\cC$. Therefore, by the estimate of the computational
complexity of \texttt{MIN\_REALZN} given in 
Section~\ref{new_construct_section}, for an $[n,k]$ code $\cC \in \mfC$, 
and a $(T,\omega) \in \cQ(\cC)$ such that $\k(\cC;T,\omega) \leq t$, 
the minimal realization $\cM(\cC;T,\omega)$ may be constructed in 
$O(k^2n^2)$ time. Note that $t$ appears in the exponent of the
constant implicit in the $O$-notation.

This leaves us with the problem of finding, for a given code $\cC \in \TW(t)$,
a tree decomposition $(T,\omega) \in \cQ(\cC)$ such that 
$\k(\cC;T,\omega) \leq t$. Unfortunately, there appears to be no 
efficient algorithm known for solving this problem. However, reasonably
good algorithms do exist for solving a closely related problem: 
given a code $\cC \in \BW(t)$, find a tree decomposition 
$(T,\omega) \in \cQ(\cC)$ such that $\sg(\cC;T,\omega) \leq t$. 
Several polynomial-time algorithms for solving this problem are given
in \cite{HO07}, the most efficient of these being an algorithm 
that runs in $O(n^3)$ time\footnote{As usual, the constant hidden in the 
$O$-notation depends exponentially on $t$.}, $n$ being the length of $\cC$.
Now, by Proposition~\ref{bw_tw_bnds}, any $\cC \in \TW(t)$ is also
in $\BW(t)$, and furthermore, $\k(\cC;T,\omega) \leq 2\sg(\cC;T,\omega)$. 
Therefore, the algorithms of \cite{HO07} find, for a given code 
$\cC \in \TW(t)$, a tree decomposition $(T,\omega) \in \cQ(\cC)$ 
such that $\k(\cC;T,\omega) \leq 2t$. 
This is sufficient for our purposes, as the computational complexity
of the sum-product algorithm on the resulting $\cM(\cC;T,\omega)$
would still be $O(n)$, except that the constant in the $O$-notation
would now be proportional to $3q^{2t}$.

While code families of bounded treewidth have the desirable property
of having linear decoding complexity, it is very likely that they
do not have good error-correcting properties. We give an argument to
support the plausibility of this statement. Recall from coding theory 
that a code family $\mfC$ is called \emph{asymptotically good} if there 
exists a sequence of $[n_i,k_i,d_i]$ codes $\cC_i \in \mfC$, with 
$\lim_i n_i = \infty$, such that $\liminf_i k_i/n_i$ and 
$\liminf_i d_i/n_i$ are both strictly positive. The code family 
$\cC_i$, $i \geq 1$, from the proof of Corollary~\ref{pathwidth_cor} 
has bounded treewidth, but is not asymptotically good: 
$k_i/n_i \rightarrow 1/4$, but $d_i/n_i \rightarrow 0$, 
as $i \rightarrow \infty$. 

It is known that if a code family $\mfC$ has bounded trellis complexity, 
\emph{i.e.}, if there exists an integer $t$ such that 
$\sg_\tr(\cC) \leq t$ for all $\cC \in \mfC$, then $\mfC$ is not 
asymptotically good. This is a consequence of the following bound 
for an $[n,k,d]$ code $\cC$ \cite{LV95}:
\beq
\sg_\tr(\cC) \geq \frac{k}{n} \, (d-1).
\label{s_trellis_bnd}
\eeq
Since $\sg_\tr(\cC) \leq \k_\tr(\cC)$, the above is also a lower bound
on $\k_\tr(\cC)$. 

Now, suppose that $\cC_i$, $i \geq 1$, is a sequence of $[n_i,k_i,d_i]$
codes of bounded treewidth, so that there exists some $t \geq 0$
such that for all $i$, $\k(\cC_i) \leq t$. Hence, from (\ref{s_trellis_bnd}),
we have $\k_\tr(\cC_i) - \k(\cC_i) \geq \frac{k_i}{n_i} \, (d_i-1) - t$.
Therefore, assuming the validity of Conjecture~\ref{conj1}, we have 
$$
\limsup_i \frac{\frac{k_i}{n_i} \, (d_i-1) - t}{\log n_i} < \infty.
$$
Since $t$ is a constant, this implies that 
$$
\limsup_i \frac{k_i d_i}{n_i \log n_i} < \infty,
$$
from which we infer that at least one of $\liminf_i k_i/n_i$ and 
$\liminf_i d_i/n_i$ must be zero. 

Observe that the final conclusion of the above argument can also be reached
if we only assume that $\k(\cC_i)$ grows at most logarithmically with $n_i$.
We formalize this as a conjecture.
\begin{conjecture}
Let $\cC_i$, $i \geq 1$, be any sequence of $[n_i,k_i,d_i]$ codes 
such that $\limsup_i \frac{\k(\cC_i)}{\log n_i} < \infty$.
Then, $\liminf_i k_i/n_i$ and $\liminf_i d_i/n_i$ cannot both be
strictly positive. In particular, for any $t \geq 0$, 
the code family $\TW(t)$ is not asymptotically good.
\label{conj2}
\end{conjecture}

We wrap up our discussion on complexity measures for codes 
by elaborating on a comment we made at the beginning of this 
subsection, in which we implied that hard coding-theoretic problems 
often become polynomial-time solvable when restricted to codes of bounded 
complexity. We saw earlier several examples of algorithms 
that, given a code $\cC \in \TW(t)$, solve some problem in time
polynomial in the length of $\cC$. In each of these cases, the 
computational complexity of the algorithm displayed an exponential 
dependence on the parameter $t$. But since $t$ was a fixed constant,
this exponential dependence could be absorbed into the constant 
hidden in the ``big-$O$'' estimate of the complexity. Thus, fixing
the parameter $t$ allowed a potentially intractable coding-theoretic 
problem to become tractable. Problems that may be hard in general, 
but which become solvable in polynomial time when one of the parameters 
of the problem is fixed, are called \emph{fixed-parameter tractable}. 
We noted previously that the problems of computing the treewidth and
branchwidth of a code are NP-hard. It should come as no surprise 
that these problems are in fact fixed-parameter tractable. 
Hlin{\v e}n\'y \cite{hlineny05} gives an $O(n^3)$ algorithm that,
for a fixed integer $t$, determines whether or not a given length-$n$
code is in $\BW(t)$. From this, one can also prove the existence of 
an $O(n^3)$ algorithm for deciding membership of a given 
length-$n$ code in $\TW(t)$ \cite{hlineny07}.

\section{Concluding Remarks\label{conclusion}}

Perhaps the most significant problem that remains open in the context of
minimal tree realizations of codes is the resolution of Conjecture~\ref{conj2},
which proposes that codes of bounded treewidth cannot be asymptotically good.
It may be possible to resolve this by deriving a lower bound on 
treewidth along the lines of the bound in (\ref{s_trellis_bnd}).
Another possibility is an inductive approach using code decompositions.

However, an open problem of far greater significance is the development of a
general theory of minimal realizations of codes on graphs with cycles. 
At present, such a theory only exists for the case of realizations
of codes on graphs consisting of a single cycle, \emph{i.e.}, 
tail-biting trellis realizations \cite{KV03}. This simplest case of
graphs with cycles is already more difficult to study than the cycle-free 
case --- for example, there can be several non-equivalent definitions
of minimality in the context of tail-biting trellis realizations. 
The challenge posed by graphs with more complex cycle structures can
only be greater.

\appendix

\section{Proofs of Lemmas~\ref{b|e_lemma} and 
\ref{ess_lemma}\label{sec2_lemmas_app}}

\noindent \emph{Proof of Lemma~\ref{b|e_lemma}}.
Consider an arbitrary $e \in E$. 
An arbitrary global configuration $\b$ may be written in the form
$({\b|}_{J(e)},{\b|}_{E(T_e)},{\b|}_e,
{\b|}_{E(\bar{T}_e)},{\b|}_{\bar{J}(e)})$.
Now, suppose that $\b$ is such that ${\b|}_e = \0$, \emph{i.e.},
$\b = ({\b|}_{J(e)},{\b|}_{E(T_e)},\0,
{\b|}_{E(\bar{T}_e)},{\b|}_{\bar{J}(e)})$.
Observe that the global configurations 
$$
\b' = ({\b|}_{J(e)},{\b|}_{E(T_e)},\0,\0,\0) \ \ \text{ and } \ \
\b'' = (\0,\0,\0,{\b|}_{E(\bar{T}_e)},{\b|}_{\bar{J}(e)})
$$
also satisfy all local constraints (since $\0 \in {\mfB|}_v$ 
for each $v \in V$), and hence are in $\mfB$.
Therefore, $({\b|}_{J(e)},\0) = {\b'|}_I \in \cC$, and
so by definition of $\cC_{J(e)}$, we have 
${\b|}_{J(e)} \in \cC_{J(e)}$. Similarly, 
$(\0,{\b|}_{\bar{J}(e)}) = {\b''|}_I \in \cC$, so that 
${\b|}_{\bar{J}(e)} \in \cC_{\bar{J}(e)}$. 
Hence, ${\b|}_I = ({\b|}_{J(e)},{\b|}_{\bar{J}(e)}) 
\in \cC_{J(e)} \oplus \cC_{\bar{J}(e)}$. 
\qed
\mbox{}\\

\noindent \emph{Proof of Lemma~\ref{ess_lemma}}.
For any tree model (essential or not), we have, by definition, 
${\mfB|}_v \subseteq C_v$ for all $v \in V$. So we need
only show the reverse inclusion in the case when 
$\G=(T,\omega,\,(\cS_e,\, e \in E),\, (C_v,\, v \in V))$ is an essential 
tree model.

Pick an arbitrary $v \in V$. Let $e_1,e_2,\ldots,e_\delta$ be the edges
of $T$ incident with $V$. For $i = 1,2,\ldots,\delta$, let $T_i$ denote
the component of $T - e_i$ that does not include $v$. Set $F_i = E(T_i)$,
and $J_i = \omega^{-1}(V(T_i))$.
We will write an arbitrary configuration $\b \in \mfB$ as 
$$
\left({\b|}_{\omega^{-1}(v)},\, ({\b|}_{e_1}, {\b|}_{F_1}, {\b|}_{J_1}),
\ldots,({\b|}_{e_\delta}, {\b|}_{F_\delta}, {\b|}_{J_\delta})\right).
$$

Consider any $(\c_0,\c_1,\ldots,\c_\delta) \in C_v$,
where $\c_0 \in \F^{\omega^{-1}(v)}$, and 
$\c_i \in \cS_{e_i}$ for $i = 1,\ldots,\delta$. As the tree model $\G$ is 
essential, we have $\cS_{e_i} = {\mfB|}_{e_i}$ for all $i$. 
In particular, $\c_i \in {\mfB|}_{e_i}$, so that there exists 
$\b^{(i)}  \in \mfB$ such that ${\b^{(i)}|}_{e_i} = \c_i$.
As $\b^{(i)}$ is in $\mfB$, its ``sub-configuration'' 
$(\c_i,{\b^{(i)}|}_{F_i},{\b^{(i)}|}_{J_i})$ satisfies the local
constraints of $\G$ at all vertices in $V(T_i)$. Hence, 
$$
\bar{\b} \ = \ \left(\c_0,\, (\c_1,{\b^{(1)}|}_{F_1},{\b^{(1)}|}_{J_1}),
\ldots, (\c_\delta,{\b^{(\delta)}|}_{F_\delta},{\b^{(\delta)}|}_{J_\delta})
\right)
$$
satisfies the local constraints of $\G$ at all vertices in 
$\bigcup_{i=1}^{\delta} V(T_i)$. Now, $v$ is the only vertex of $T$ 
that is not in $\bigcup_{i=1}^{\delta} V(T_i)$. But, by construction,
${\bar{\b}|}_v = (\c_0,\c_1,\ldots,\c_\delta) \in C_v$,
and so, $\bar{\b}$ also satisfies the local constraint at $v$. Thus,
$\bar{\b}$ satisfies all local constraints of $\G$, so that
$\bar{\b} \in \mfB$. Hence, 
$(\c_0,\c_1,\ldots,\c_\delta) = {\bar{\b}|}_v$ is in ${\mfB|}_v$,
which proves the lemma.
\qed

\section{Proofs of Lemmas~\ref{Gbar_lemma1} and 
\ref{Gbar_lemma2}\label{Gbar_lemmas_app}}

\noindent \emph{Proof of Lemma~\ref{Gbar_lemma1}}. 
For simplicity of notation, let $F$ denote the edge set of the subtree 
$T_\he$, and let $\bar{F}$ denote that of the subtree $\bar{T}_\he$.
Note that $F \cup \bar{F} = E(T) - {\he}$.
Throughout this proof, we will write an arbitrary global configuration $\b$,
belonging to $\mfB$ or $\mfB(\bar{\G})$, in the form 
$({\b|}_J,{\b|}_F,{\b|}_{\he},{\b|}_{\bar{F}},{\b|}_{\bar{J}})$.

Consider any $\ob = ({\ob|}_J,{\ob|}_F,{\ob|}_\he,{\ob|}_\oF,{\ob|}_\oJ) 
\in \mfB(\bar{\G})$. Let $\ell$ and $r$ be the two vertices incident 
with the edge $\he$ in $T$. We assume that $\ell \in V(T_\he)$ and 
$r \in V(\oT_\he)$, as depicted in Figure~\ref{lemma1_pf_fig}. 
We write the local configuration ${\ob|}_\ell$ as 
$({\ob|}_{E(\ell) - \he}, {\ob|}_{\omega^{-1}(\ell)}, {\ob|}_\he)$,
and ${\ob|}_{r}$ as
$({\ob|}_\he, {\ob|}_{\omega^{-1}(r)}, {\ob|}_{E(r) - \he})$.

Suppose first that ${\ob|}_\he = \0$; note that the zero element of 
$\bar{\cS}_\he$ ($=\cS_\he/W$) is $W$. By definition of $\bar{\G}$,
${\ob|}_\ell \in {\bar{\mfB}|}_\ell = {\Phi(\mfB)|}_\ell$. Hence,
there exists $({\ob|}_{E(\ell) - \he}, {\ob|}_{\omega^{-1}(\ell)}, \w)
\in {\mfB|}_\ell$, for some $\w \in W$. Now, $({\ob|}_J,{\ob|}_F)$ 
(being a ``sub-configuration'' of $\ob$) satisfies the local
constraints of $\bar{\G}$ at all vertices in $V(T_\he) - \{\ell\}$. But these
local constraints are of the form ${\bar{\mfB}|}_v$ which, for
$v \in V(T_\he) - \{\ell\}$, is identical to $\mfB_v$. Therefore, 
the sub-configuration $({\ob|}_J,{\ob|}_F)$ satisfies the local
constraints of $\G$ at all vertices in $V(T_\he) - \{\ell\}$. It 
follows that $({\ob|}_J,{\ob|}_F,\w)$ satisfies the local constraints 
of $\G$ at all vertices in $V(T_\he)$, including $\ell$. By a similar argument,
there exists a $\w' \in W$ such that $(\w',{\ob|}_\oF,{\ob|}_\oJ)$ satisfies
the local constraints of $\G$ at all vertices in $V(\oT_\he)$.

\begin{figure}
\epsfig{file=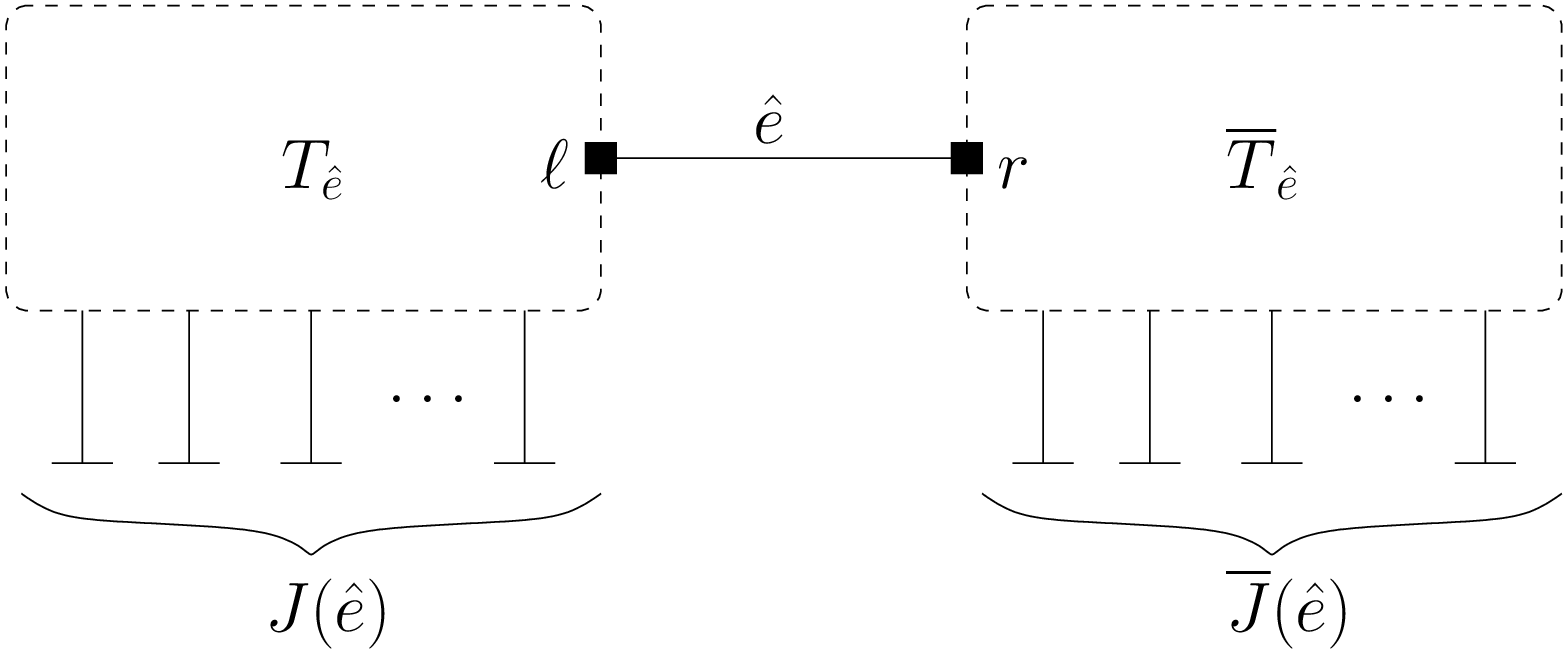, width=8cm}
\caption{A depiction of the two subtrees $T_\he$ and $\oT_\he$
connected by the edge $\he$.}
\label{lemma1_pf_fig}
\end{figure}

Now, by definition of $W$, there exist $\b$ and $\b'$ in $\mfB$, such that
$\b = ({\b|}_J,{\b|}_F,\w,{\b|}_\oF,{\b|}_\oJ)$,
$\b' = ({\b'|}_J,{\b'|}_F,\w',{\b'|}_\oF,{\b'|}_\oJ)$,
and 
$({\b|}_J,{\b|}_\oJ), ({\b'|}_J,{\b'|}_\oJ) \in \cC_J \oplus \cC_\oJ$.
Note, in particular, that the sub-configuration
$(\w,{\b|}_\oF,{\b|}_\oJ)$ 
of $\b$ satisfies the local constraints of $\G$ at all vertices in 
$V(\oT_\he)$. Therefore, the global configuration 
$\g = ({\ob|}_J,{\ob|}_F,\w,{\b|}_\oF,{\b|}_\oJ)$ satisfies the local 
constraints of $\G$ at all vertices in $T$, and hence is in the
full behavior, $\mfB$, of $\G$. A similar argument shows that 
$\g' = ({\b'|}_J,{\b'|}_F,\w',{\ob|}_\oF,{\ob|}_\oJ)$ is also in $\mfB$.

As $\mfB$ is a vector space, it must also contain 
$$
\b-\g = ({\b|}_J - {\ob|}_J, {\b|}_F - {\ob|}_F, \0,\0,\0)
$$
and
$$
\b'-\g' = (\0,\0,\0, {\b'|}_\oF - {\ob|}_\oF, {\b'|}_\oJ - {\ob|}_\oJ).
$$
Since $\G$ is a tree realization of $\cC$, we have ${\mfB|}_I = \cC$.
In particular, $({\b|}_J - {\ob|}_J, \0) = {(\b-\g)|}_I \in \cC$,
and similarly, $(\0, {\b'|}_\oJ - {\ob|}_\oJ) = {(\b'-\g')|}_I \in \cC$.
Hence, ${\b|}_J - {\ob|}_J \in \cC_J$ and 
${\b'|}_\oJ - {\ob|}_\oJ \in \cC_\oJ$. However, $\b$ and $\b'$ were
chosen so that ${\b|}_J \in \cC_J$ and ${\b'|}_\oJ \in \cC_\oJ$.
Thus, we also have ${\ob|}_J \in \cC_J$ and ${\ob|}_\oJ \in \cC_\oJ$.
This finally yields 
${\ob|}_I = ({\ob|}_J, {\ob|}_\oJ) \in \cC_J \oplus \cC_\oJ$,
thus proving one direction of part (b) of the lemma.

We will next show that if ${\ob|}_\he \neq \0$, then $\ob_I \in \cC$
but $\ob_I \notin \cC_J \oplus \cC_\oJ$. This will prove 
both part (a) and the reverse direction of part (b).

So, suppose that ${\ob|}_\he = \bar{\s} \neq \0$. Thus, $\bar{\s}$
is some coset of $W$ in $\cS_\he$, but is not $W$ itself. Pick some
$\s \in \bar{\s}$. As $\cS_\he = {\mfB|}_\he$, there exists some
$\b \in \mfB$ such that ${\b|}_\he = \s$. Observe that 
${\b|}_I \in {\mfB|}_I = \cC$, but since $\s \in \bar{\s} \neq W$,
${\b|}_I \notin \cC_J \oplus \cC_\oJ$. 

Define $\tb = \Phi(\b)$, so that $\tb \in \bar{\mfB}$.
Furthermore, ${\tb|}_\he = \bar{\s}$, and 
${\tb|}_I$ ($= {\b|}_I$) is in $\cC$ but not in $\cC_J \oplus \cC_\oJ$.
We have already noted (prior to the statement of Lemma~\ref{Gbar_lemma1})
that $\bar{\mfB} \subseteq \mfB(\bar{\G})$. Therefore, 
$\tb \in \mfB(\bar{\G})$, and since $\mfB(\bar{\G})$ is a vector space, 
$\ob - \tb \in \mfB(\bar{\G})$. 

However, ${(\ob - \tb)|}_\he = \bar{\s} - \bar{\s} = \0$, and as we
showed above, this implies that ${(\ob - \tb)|}_I \in \cC_J \oplus \cC_\oJ$.
Since ${\tb|}_I$ is in $\cC$ but not in $\cC_J \oplus \cC_\oJ$,
we find that ${\ob|}_I \in \cC$, but ${\ob|}_I \notin \cC_J \oplus \cC_\oJ$.

The proof of the lemma is now complete. \qed
\mbox{}\\

\noindent \emph{Proof of Lemma~\ref{Gbar_lemma2}}. 
As $\bar{\G}$ is a tree realization of $\cC$, Lemma~\ref{b|e_lemma} shows
that for any $\ob \in \mfB(\bar{\G})$, we have ${\ob|}_{e'} = \0$ 
only if ${\ob|}_I \in \cC_{J(e')} \oplus \cC_{\bar{J}(e')}$. Thus, we
need only prove the converse.

Suppose that $\ob \in \mfB(\bar{\G})$ is such that 
$({\ob|}_{J(e')},{\ob|}_{\oJ(e')}) \in \cC_{J(e')} \oplus \cC_{\bar{J}(e')}$,
but ${\ob|}_{e'} \neq \0$.
Now, ${\ob|}_{e'} \in \bar{\cS}_{e'} = {\bar{\mfB}|}_{e'} = 
{{\mfB}|}_{e'}$, the last equality being a consequence of the fact that
$e' \neq \he$. Therefore, there exists a $\b \in \mfB$ such that 
${\b|}_{e'} = {\ob|}_{e'}$. Note that, by the hypothesis of the lemma, 
${\b|}_I \notin \cC_{J(e')} \oplus \cC_{\bar{J}(e')}$. 

Set $\tb = \Phi(\b)$, so that $\tb \in \bar{\mfB} \subseteq \mfB(\bar{\G})$.
Observe that ${\tb|}_{I} = {\b|}_{I}$, and 
since $e' \neq \he$, we also have ${\tb|}_{e'} = {\b|}_{e'}$. 
Thus, ${\tb|}_{e'} = {\ob|}_{e'}$ and 
${\tb|}_I \notin \cC_{J(e')} \oplus \cC_{\bar{J}(e')}$.
But now, we have $\ob - \tb \in \mfB(\bar{\G})$, with 
${(\ob - \tb)|}_{e'} = \0$, and ${(\ob - \tb)|}_{e'} \notin 
\cC_{J(e')} \oplus \cC_{\bar{J}(e')}$. This contradiction 
of Lemma~\ref{b|e_lemma} proves that there exists no
$\ob \in \mfB(\bar{\G})$ such that 
$({\ob|}_{J(e')},{\ob|}_{\oJ(e')}) \in \cC_{J(e')} \oplus \cC_{\bar{J}(e')}$,
but ${\ob|}_{e'} \neq \0$. \qed

\section{Proof of Forward Direction of 
Theorem~\ref{rsum_theorem}\label{rsum_app}}

\noindent \emph{Proof of \emph{(a)} $\Rightarrow$ \emph{(b)} in 
Theorem~\ref{rsum_theorem}}. Let $\cC = \cC_1 \oplus_r \cC_2$ for 
codes $\cC_1$ and $\cC_2$ defined on the index sets $I_1$ and $I_2$, 
respectively. By definition, $I = I_1 \D I_2$. 
Set $J = I_1 - I_2$ and $\oJ = I_2 - I_1$, so that $(J,\oJ)$ forms a 
partition of $I$. In what follows, words defined on the index set $I_1$ 
will be written in the form $\x = ({\x|}_J, {\x|}_{I_1 \cap I_2})$;
words defined on the index set $I_2$ will be written in the form 
$\x = ({\x|}_{I_1 \cap I_2}, {\x|}_{\oJ})$; words
defined on the index set $I$ will be written as 
$\x = ({\x|}_J,{\x|}_{\oJ})$; and finally, words
on the index set $I_1 \cup I_2$ will be written as 
$\x = ({\x|}_J,{\x|}_{I_1 \cap I_2},{\x|}_{\oJ})$.

We begin by proving that $\dim({\cC|}_J) = \dim(\cC_1)$. This is 
accomplished by a two-step argument: we first show that 
${\cC|}_J = {\cC_1|}_J$, and then we show that ${\cC_1|}_J \cong \cC_1$.

If $\a \in {\cC|}_J$, then there exists some $\x \in \cC_1$, $\y \in \cC_2$
such that ${(\x \star \y)|}_J = \a$. However, ${(\x \star \y)|}_J = {\x|}_J$
as $J$ lies outside $I_1 \cap I_2$. Hence, $\a = {\x|}_J \in {\cC_1|}_J$.
Conversely, suppose that $\a \in {\cC_1|}_J$. Then, there exists 
$\z \in \cC_1^{(p)}$ such that $\x = (\a,\z) \in \cC_1$. Since 
$\cC_1^{(p)} = \cC_2^{(p)}$, there exists $\y = (\z,\b) \in \cC_2$.
Now, $\x \star \y = (\a,\0,\b)$, and hence $(\a,\b) \in \cC$. Thus,
$\a \in {\cC|}_J$, which completes the proof of the fact that 
${\cC|}_J = {\cC_1|}_J$.

Now, to show that $\cC_1 \cong {\cC_1|}_J$, let us consider
the projection map $\pi: \cC_1 \rightarrow {\cC_1|}_J$ defined
by $\pi(\x) = {\x|}_J$. This map is a homomorphism, with kernel
isomorphic to $\cC_1^{(s)}$, which is $\{\0\}$ by definition. 
Hence, $\pi$ is in fact an isomorphism, which proves that 
$\cC_1 \cong {\cC_1|}_J$.

We have thus shown that $\dim({\cC|}_J) = \dim(\cC_1)$. A similar argument
yields the fact that $\dim({\cC|}_\oJ) = \dim(\cC_2)$. Hence,
$$
 \dim({\cC|}_J) + \dim({\cC|}_\oJ) - \dim(\cC) 
 = \dim(\cC_1) + \dim(\cC_2) - \dim(\cC_1 \oplus_r \cC_2)
 = r,
$$
by Corollary~\ref{dim_rsum_cor}.

It remains to show that $\min\{|J|,|\oJ|\} \geq r$. Note that
since $\dim({\cC|}_J) + \dim({\cC|}_\oJ) - \dim(\cC) = r$, and 
$\dim({\cC|}_\oJ) \leq \dim(\cC)$, we must have 
$\dim({\cC|}_J) \geq r$. Therefore, $|J| \geq \dim({\cC|}_J) \geq r$.
By a similar argument, we also have $|\oJ| \geq r$. \qed

\section{Proof of Proposition~\ref{rec_constr_prop}\label{rec_constr_app}}

The proof of Proposition~\ref{rec_constr_prop} requires the following
lemma, which presents a property of the codes $\cC_1$ and $\cC_2$ obtained
via the $r$-sum decomposition procedure of Section~\ref{decomp_section}.

\begin{lemma}
Let $\cC$ be a code defined on the index set $I$, and let 
$(J,\oJ)$ be a partition of $I$, with 
$\dim({\cC|}_J) + \dim({\cC|}_\oJ) - \dim(\cC) = r$.
Suppose that $\cC_1$ and $\cC_2$ are the codes, defined on the
respective index sets $I_1$ and $I_2$, that are obtained by the 
procedure described in the proof of Theorem~\ref{rsum_theorem}. 
Then, for any $J_1 \subseteq J$, and any $J_2 \subseteq \oJ$, we have 
\begin{eqnarray}
\dim({\cC_1|}_{J_1}) + \dim({\cC_1|}_{I_1 - J_1}) - \dim(\cC_1)
&=& \dim({\cC|}_{J_1}) + \dim({\cC|}_{I-J_1}) - \dim(\cC), 
\label{appD_lemma_eq1}
\\
\dim({\cC_2|}_{J_2}) + \dim({\cC_2|}_{I_2 - J_2}) - \dim(\cC_2)
&=& \dim({\cC|}_{J_2}) + \dim({\cC|}_{I-J_2}) - \dim(\cC).
\label{appD_lemma_eq2}
\end{eqnarray}
\label{appD_lemma}
\end{lemma}
\begin{proof}
We use notation from the proof of the (b)~$\Rightarrow$~(a) direction
of Theorem~\ref{rsum_theorem}. Thus, $\cC$, $\cC_1$ and $\cC_2$
are generated by the matrices $\bar{G}$, $G_1$ and $G_2$ 
given by (\ref{rref_eq}), (\ref{G1_def}) and (\ref{G2_def}), 
respectively, which we reproduce here for the sake of convenience. 
\begin{eqnarray*}
\overline{G} &=&
\left[
\begin{array}{cccc}
I_{k_1} & A & \O & B \\
\O & \O & I_{k-k_1} & C
\end{array}
\right], \\
G_1 &=& \left[\begin{array}{ccc}
I_{k_1} & A & X \\
\end{array}
\right], \\
G_2 &=& \left[\begin{array}{ccc} X & \O & B \\ 
\O & I_{k-k_1} & C \end{array}\right].
\end{eqnarray*}
For any matrix $M$, given a subset $Z$ of the column indices of $M$, we 
will denote by ${M|}_Z$ the restriction of $M$ to the columns indexed
by $Z$. Thus,
$$
{\bar{G}|}_{J} = 
\left[
\begin{array}{cc}
I_{k_1} & A \\
\O & \O 
\end{array}
\right],\ \ \ \ 
{G_1|}_J = 
\left[
\begin{array}{cc}
I_{k_1} & A 
\end{array}
\right],
$$
$$
{\bar{G}|}_{\oJ} = {G_2|}_\oJ = 
\left[
\begin{array}{cc}
\O & B \\
I_{k-k_1} & C
\end{array}
\right].
$$

Our proof of the lemma uses only elementary linear algebra.
We prove (\ref{appD_lemma_eq1}) first.
Consider any $J_1 \subseteq J$. It is clear that
$
{\bar{G}|}_{J_1} = 
\left[
\begin{array}{c}
{G_1|}_{J_1} \\
\O
\end{array}
\right]
$, 
and therefore, we have $\dim({\cC|}_{J_1}) = \rank({\bar{G}|}_{J_1})
= \rank({G_1|}_{J_1} ) 
= \dim({\cC_1|}_{J_1})$.
Next, note that $I - J_1 = (J-J_1) \disj \oJ$,\, from which we have 
$$
\dim({\cC|}_{I-J_1}) = 
\rank\left(\left[\ {\bar{G}|}_{J - J_1} \ \ \ {\bar{G}|}_{\oJ}\ \right]\right).
$$
Now, observe that by performing column operations on ${\bar{G}|}_{\oJ}$,
we can bring it into the form 
$$
W = 
\left[
\begin{array}{cc}
\O & B \\
I_{k-k_1} & \O
\end{array}
\right].
$$
Hence,
\begin{eqnarray*}
\rank\left(\left[\ {\bar{G}|}_{J - J_1} \ \ \ {\bar{G}|}_{\oJ}\ \right]\right)
&=& \rank\left(\left[\ {\bar{G}|}_{J - J_1} \ \ \ W\ \right]\right) \\
&=& \rank(I_{k-k_1}) + 
\rank\left(\left[\ {G_1|}_{J - J_1} \ \ \ B\ \right]\right) \\
&=& k-k_1 + \rank\left(\left[\ {G_1|}_{J - J_1} \ \ \ B\ \right]\right).
\end{eqnarray*}
At this point, we have
$$
\dim({\cC|}_{J_1}) + \dim({\cC|}_{I-J_1}) = \dim({\cC_1|}_{J_1}) + 
k-k_1 + \rank\left(\left[\ {G_1|}_{J - J_1} \ \ \ B\ \right]\right),
$$
which upon re-arrangement yields
$$
\dim({\cC|}_{J_1}) + \dim({\cC|}_{I-J_1}) - \dim(\cC)
= \dim({\cC_1|}_{J_1}) + 
\rank\left(\left[\ {G_1|}_{J - J_1} \ \ \ B\ \right]\right) - 
\dim(\cC_1).
$$
Thus, (\ref{appD_lemma_eq1}) would be proved if we could establish
that $\dim({\cC_1|}_{I_1-J_1}) = 
\rank\left(\left[\ {G_1|}_{J - J_1} \ \ \ B\ \right]\right)$.

Now, $I_1 - J_1 = (J-J_1) \disj I_X$, and hence,
$$
\dim({\cC_1|}_{I_1-J_1}) = 
\rank\left(\left[\ {G_1|}_{J - J_1} \ \ \ X\ \right]\right).
$$
Thus, we have to show that 
$
\rank\left(\left[\ {G_1|}_{J - J_1} \ \ \ B\ \right]\right)
= \rank\left(\left[\ {G_1|}_{J - J_1} \ \ \ X\ \right]\right).
$
We will prove that the matrices $B$ and $X$ have identical column-spaces. 
Clearly, the desired result then follows.

Recall that for $i = 1,2,\ldots,k_1$, the $i$th row of $B$ can be 
uniquely expressed as a linear combination, $\sum_{j=1}^r \alpha_{i,j} \b_j$, 
of its first $r$ rows $\b_1,\ldots,\b_r$. Furthermore, the $i$th row of 
$X$ equals $\sum_{j=1}^r \alpha_{i,j} \d_j$ for the same $\alpha_{i,j}$'s,
where $\d_1,\ldots,\d_r$ are the rows of the generator matrix, $D_r$, of
the code $\D_r$. In particular, the first $r$ rows of $X$ constitute the
matrix $D_r$. Denote by $B_r$ the submatrix of $B$ comprised by
its first $r$ rows. 

Now, it was pointed out in Section~\ref{decomp_section}
(a little after the proof of Proposition~\ref{dim_sum_prop}) that 
any column vector in $\F_q^r$ is a scalar multiple of some column of $D_r$.
Therefore, any column of $B_r$ is a scalar multiple of some column of $D_r$.
But because of the way $X$ was constructed, this implies that any column
of $B$ is a scalar multiple of some column of $X$. Thus, the column-space
of $B$ is a subspace of the column-space of $X$. However, we also
have $\rank(B) = \rank(X)$, and so, the column-spaces of the two matrices
are in fact identical. This proves that 
$\rank\left(\left[\ {G_1|}_{J - J_1} \ \ \ B\ \right]\right)
= \rank\left(\left[\ {G_1|}_{J - J_1} \ \ \ X\ \right]\right)$,
and (\ref{appD_lemma_eq1}) follows. \\[-4pt]

To show (\ref{appD_lemma_eq2}), consider any $J_2 \subseteq \oJ$. 
Arguments similar to the ones above establish that 
\beq
\dim({\cC|}_{J_2}) = \dim({\cC_2|}_{J_2})
\ \ \text{ and } \ \ 
\dim({\cC|}_{I-J_2}) = k_1 + 
\rank\left({\left[I_{k-k_1}\ \ \ C\right]\bigm|}_{\oJ - J_2}\right).
\label{appD_eq1}
\eeq
Now, consider $\dim({\cC_2|}_{I_2-J_2}) = \rank({G_2|}_{I_2-J_2})$.
Noting that $I_2 - J_2 = I_X \disj (\oJ-J_2)$, we see that
the matrix ${G_2|}_{I_2-J_2}$ has the form
$$
\left[
\begin{array}{ccc}
X & \O & {B|}_K \\
\O & I' & {C|}_K
\end{array}
\right],
$$
with $I' = {(I_{k-k_1})|}_{\oJ-J_2-K}$, for some $K \subseteq \oJ-J_2$.
Since the columns of $B$ are contained in the column-space of $X$,
we can perform column operations on ${G_2|}_{I_2-J_2}$ to bring it
into the form 
$$
W' = \left[
\begin{array}{ccc}
X & \O & \O \\
\O & I' & {C|}_K
\end{array}
\right].
$$
Hence, 
\begin{eqnarray}
\rank({G_2|}_{I_2-J_2}) \ \ =\ \ \rank(W') 
&=& \rank(X) + \rank([I'\ \ \ {C|}_K]) \notag \\
&=& (k_1+k_2-k) + 
\rank\left({\left[I_{k-k_1}\ \ \ C\right]\bigm|}_{\oJ - J_2}\right). 
\label{appD_eq2}
\end{eqnarray}
Some trivial manipulations of (\ref{appD_eq1}) and (\ref{appD_eq2})
yield (\ref{appD_lemma_eq2}), which proves the lemma.
\end{proof}
\mbox{}

\noindent \emph{Proof of Proposition~\ref{rec_constr_prop}}.
Recall that 
$\G^* = \left(T,\omega,\,(\cS_e,\, e \in E(T)),\,(C_v,\, v \in V(T))\right)$,
where $\cS_e$ and $\C_v$ are as defined in (\ref{Se_def}) and (\ref{Cv_def}).
To show that $\G^*$ is the minimal realization $\cM(\cC;T,\omega)$,
it is enough to show that for all $e \in E(T)$, $\dim(\cS_e)$
equals the expression in (\ref{dimSe*_alt}), \emph{i.e.},
\beq
\dim(\cS_e) = \dim({\cC|}_{J(e)}) + \dim({\cC|}_{\bar{J}(e)}) - \dim(\cC).
\label{appD_eq3}
\eeq
Note that this is true when $e = \he$, since $\dim(\cS_\he) = \dim(\D_r)
= r$, and from (\ref{r_def}), we have 
$r = \dim({\cC|}_{J(\he)}) + \dim({\cC|}_{\oJ(\he)}) - \dim(\cC)$.
We must therefore show that (\ref{appD_eq3}) holds for 
$e \in E(T) - \{\he\} = E(T_\he) \cup E(\oT_\he)$. We will
prove this for $e \in E(T_\he)$; the proof for $e \in E(\oT_\he)$ is similar.

\begin{figure}
\epsfig{file=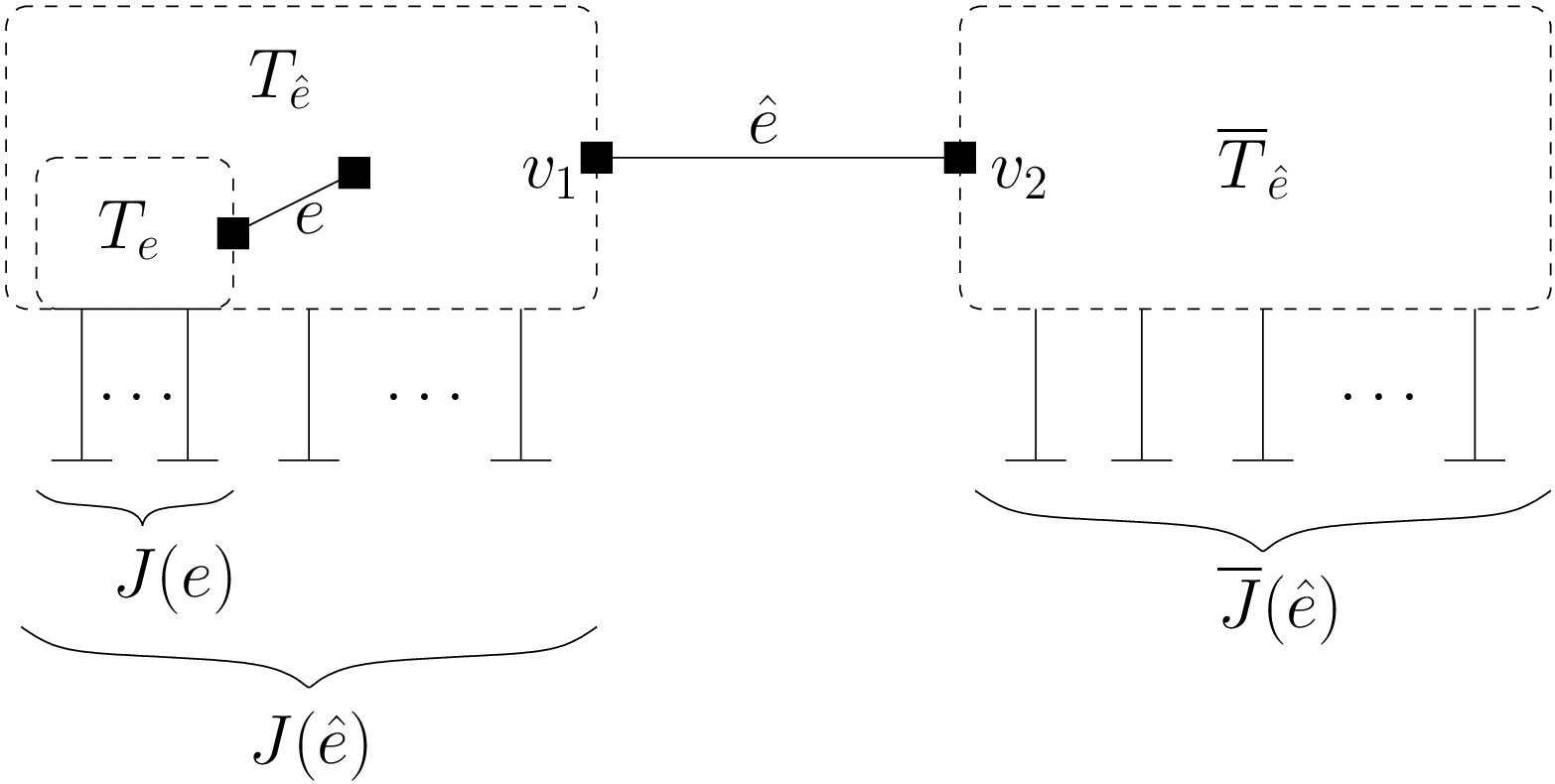, width=9.5cm}
\caption{$T_e$ is a subtree of $T_\he$.}
\label{subtree}
\end{figure}

So, consider any $e \in E(T_\he)$. One of the two components, $T_e$
and $\oT_e$, of $T - e$ is contained in $T_\he$. Without loss of generality, 
we may assume that it is $T_e$ that is a subtree of $T_\he$, as depicted
in Figure~\ref{subtree}. Hence, $J(e) = \omega^{-1}(V(T_e)) \subseteq J(\he)$.
Now, by (\ref{Se_def}), $\cS_e = \cS_e^{(1)}$, the latter being the
state space associated with $e$ in $\cM(\cC_1;T_\he,\omega_1)$.
Therefore, by (\ref{dimSe*_alt}), 
$$
\dim(\cS_e^{(1)}) = 
\dim({\cC_1|}_{J(e)}) + \dim({\cC_1|}_{I_1 - J(e)}) - \dim(\cC_1).
$$
But, by Lemma~\ref{appD_lemma}, the above expression is equal to the 
expression on the right-hand side of (\ref{appD_eq3}). Hence,
(\ref{appD_eq3}) holds for any $e \in E(T_\he)$, and the 
proposition follows. \qed



\begin{thebibliography}{99}

\bibitem{AM00} S.M.\ Aji and R.J.\ McEliece, 
``The generalized distributive law,'' 
\emph{IEEE Trans.\ Inform.\ Theory}, vol.\ 46, no.\ 2, pp.\ 325--343, 2000.

\bibitem{arnborg} S.\ Arnborg, D.G.\ Corneil and A.\ Proskurowski,
``Complexity of finding embeddings in a $k$-tree,''
\emph{SIAM J.\ Alg.\ Disc.\ Meth.}, vol.\ 8, pp.\ 277--284, 1987.

\bibitem{AP89} S.\ Arnborg and A.\ Proskurowski, 
``Linear time algorithms for NP-hard problems restricted to partial k-trees,'' 
\emph{Discrete Applied Mathematics}, vol.\ 23, no.\ 1, pp.\ 11--24, 1989.

\bibitem{BMvT} E.R.\ Berlekamp, R.J.\ McEliece, and H.C.A.\ van Tilborg, 
``On the inherent intractability of certain coding problems,'' 
\emph{IEEE Trans.\ Inform.\ Theory}, vol.\ IT-24, pp.\ 384--386, 1978.

\bibitem{bod88} H.L.\ Bodlaender, 
``Dynamic programming on graphs of bounded treewidth,''
\emph{Proc.\ 15th International Colloquium on Automata, Languages and 
Programming}, vol.\ 317, Lecture Notes in Computer Science, 
Springer-Verlag, pp. 105--118, 1988.

\bibitem{bod93} H.L.\ Bodlaender, ``A tourist guide through treewidth,''
\emph{Acta Cybernetica}, vol.\ 11, pp.\ 1--23, 1993.

\bibitem{For01} G.D.\ Forney Jr.,
``Codes on graphs: normal realizations,''
\emph{IEEE Trans.\ Inform.\ Theory}, 
vol.\ 47, no.\ 2, pp.\ 520--548, Feb.\ 2001.

\bibitem{For03} G.D.\ Forney Jr.,
``Codes on graphs: constraint complexity of cycle-free
realizations of linear codes,'' 
\emph{IEEE Trans.\ Inform.\ Theory}, 
vol.\ 49, no.\ 7, pp.\ 1597--1610, July 2003.


\bibitem{halford} T.R.\ Halford and K.M.\ Chugg,
``The extraction and complexity limites of graphical models
for linear codes,'' \emph{IEEE Trans.\ Inform.\ Theory}, 
to appear.

\bibitem{HB68} S.L.\ Hakimi and J.G.\ Bredeson, ``Graph theoretic 
error-correcting codes,'' \emph{IEEE Trans.\ Inform. Theory}, 
vol.\ IT-14, pp.\ 584--591, 1968.

\bibitem{HM07} I.V.\ Hicks and N.B.\ McMurray, Jr., 
``The branch-width of graphs and their cycle matroids,''
\emph{J.\ Combin.\ Theory, Ser.\ B}, vol.\ 97, pp.\ 681--692, 2007.

\bibitem{hlineny05} P.\ Hlin{\v e}n\'y, 
``A parametrized algorithm for matroid branch-width,''
  \emph{SIAM J.\ Computing}, vol.\ 35, pp.\ 259--277, 2005.

\bibitem{hlineny07} P.\ Hlin{\v e}n\'y, personal email communication, 
Sept.\ 2007.

\bibitem{HO07} P.\ Hlin{\v e}n\'y, and S.-i.\ Oum,
``Finding branch-decompositions and rank-decompositions,'' preprint, 2007.
\texttt{http://}
\texttt{www.math.uwaterloo.ca/$\sim$sangil/pdf/2007partition.pdf}.

\bibitem{HOSG07} P.\ Hlin{\v e}n\'y, S.-i.\ Oum, D.\ Seese and G.\ Gottlob, 
``Width parameters beyond tree-width and their applications,''
\emph{The Computer Journal}, (advance access) Sept.\ 2007.
DOI 10.1093/comjnl/bxm052.

\bibitem{HW06} P.\ Hlin{\v e}n\'y and G.\ Whittle, ``Matroid tree-width,'' 
\emph{Europ.\ J.\ Combin.}, vol.\ 27, pp.\ 1117--1128, 2006.

\bibitem{kashyap} N.\ Kashyap, 
``A decomposition theory for binary linear codes,'' 
submitted to \emph{IEEE Trans.\ Inform.\ Theory}. 
ArXiv e-print cs.DM/0611028.

\bibitem{kashyap_siam} N.\ Kashyap, 
``Matroid pathwidth and code trellis complexity,'' 
\emph{SIAM J.\ Discrete Math.}, to appear.

\bibitem{KFL01} F.R.\ Kschischang, B.J.\ Frey and H.-A.\ Loeliger,
``Factor graphs and the sum-product algorithm,'' 
\emph{IEEE Trans.\ Inform.\ Theory},
vol.\ 47, no.\ 2, pp.\ 498--519, Feb.\ 2001.

\bibitem{KV03} R.\ Koetter and A.\ Vardy, 
``The structure of tail-biting trellises: minimality and basic principles,'' 
\emph{IEEE Trans.\ Inform.\ Theory}, 
vol.\ 49, no.\ 9, pp.\ 2081--2105, Sept.\ 2003.

\bibitem{LV95} A.\ Lafourcade and A.\ Vardy, 
``Asymptotically good codes have infinite trellis complexity,'' 
\emph{IEEE.\ Trans.\ Inform.\ Theory}, 
vol.\ 41, no.\ 2, pp.\ 555--559, March 1995.

\bibitem{oxley} J.G.\ Oxley, \emph{Matroid Theory}, Oxford University
Press, Oxford, UK, 1992.

\bibitem{RS-I} N.\ Robertson and P.D.\ Seymour, 
``Graph minors. I. Excluding a forest,''
\emph{J.\ Combin.\ Theory, Ser.\ B}, vol.\ 35, pp.\ 39--61, 1983.

\bibitem{RS-survey} N.\ Robertson and P.D.\ Seymour,
``Graph minors --- a survey,'' in \emph{Surveys in Combinatorics},
Cambridge University Press, 1985, pp.\ 153--171. 

\bibitem{Sey80} P.D.\ Seymour, ``Decomposition of regular matroids,''
\emph{J.\ Combin.\ Theory, Series B}, vol.\ 28, pp.\ 305--359, 1980.


\bibitem{truemper} K.\ Truemper, \emph{Matroid Decomposition}, 
Academic Press, San Diego, 1992. 

\bibitem{vanlint} J.H.\ van Lint, \emph{Introduction to Coding Theory},
3rd ed., Springer, Berlin, 1998.

\bibitem{vardy} A.\ Vardy, ``Trellis Structure of Codes,'' 
in \emph{Handbook of Coding Theory}, R.\ Brualdi, C.\ Huffman and V.\ Pless,
Eds., Amsterdam, The Netherlands: Elsevier, 1998.

\bibitem{wiberg} N.\ Wiberg, \emph{Codes and Decoding on General Graphs},
Ph.D.\ thesis, Link\"oping University, Link\"oping, Sweden, 1996.

\enlargethispage{\baselineskip}

\bibitem{WLK95} N.\ Wiberg. H.-A.\ Loeliger and R.\ Koetter, 
``Codes and iterative decoding on general graphs,'' 
\emph{Euro.\ Trans.\ Telecommun.}, vol.\ 6, pp.\ 513--525, Sept./Oct.\ 1995.

\end{thebibliography}
\end{document}